\documentclass[fleqn,thmsb,a4paper,12pt]{article}
\usepackage{amssymb}
\usepackage{amsmath}
\usepackage{color}
\usepackage{amsthm}
\usepackage{amssymb}
\usepackage{amsfonts}
\usepackage{amstext}
\usepackage[USenglish]{babel}

\newtheorem{acorollary}{Corollary}

\newtheorem{atheorem}{Theorem}
\theoremstyle{definition}
\newtheorem{adefinition}{Definition}
\newtheorem{aexample}{Example}
\newtheorem{aremark}{Remark}
\begin{document}

\author{Ata Atay\thanks{Department of Mathematical Economics, Finance and Actuarial Sciences, University of Barcelona, Spain. E-mail: aatay@ub.edu} \and Ana Mauleon\thanks{CEREC and CORE/LIDAM, UCLouvain, Belgium. E-mail: ana.mauleon@usaintlouis.be} \and Vincent Vannetelbosch\thanks{CORE/LIDAM, UCLouvain, Belgium. E-mail: vincent.vannetelbosch@uclouvain.be}}
\title{\textbf{School Choice with Farsighted Students}}
\date{November 21, 2022}
\maketitle

\begin{abstract}
We consider priority-based school choice problems with farsighted students. We show that a singleton set consisting of the matching obtained from the Top Trading Cycles (TTC) mechanism is a farsighted stable set. However, the matching obtained from the Deferred Acceptance (DA) mechanism may not belong to any farsighted stable set. Hence, the TTC mechanism provides an assignment that is not only Pareto efficient but also farsightedly stable. Moreover, looking forward three steps ahead is already sufficient for stabilizing the matching obtained from the TTC. \\

\end{abstract}

\thispagestyle{empty}\newpage

\pagenumbering{arabic}

\section{Introduction}\label{S1}

Abdulkadiro\u{g}lu and Sönmez (2003) formulate the school choice problem of assigning students to schools as a mechanism design problem.\footnote{Abdulkadiro\u{g}lu and Anderson (2022) provide an extensive survey of school choice. See also Roth and Sotomayor (1990) or Haeringer (2017) for an introduction to matching problems.} Each student has strict preferences over all schools and each school has a strict priority ordering imposed by state or local laws of all students. The outcome of a school choice problem is a matching that assigns schools to students such that each student is assigned one school and no school is assigned to more students than its capacity. Two  prominent mechanisms used for priority-based matching are the Gale and Shapley's (1962) Deferred Acceptance (DA) mechanism and the Shapley and Scarf's (1974) Top Trading Cycles (TTC) mechanism. Both mechanisms are strategy-proof: truthful preference revelation is a weakly dominant strategy for students.\footnote{Reny (2022) introduces the Priority-Efficient (PE) mechanism that always selects a Pareto efficient matching that dominates the DA stable matching, but PE is not strategy-proof. Another attempt to improve the efficiency of the DA mechanism can be found in Kesten (2010).} On the one hand, the TTC mechanism is Pareto efficient while the DA mechanism may select an inefficient matching. On the other hand, the DA mechanism is stable while the TTC mechanism may select an unstable matching. 

A stable matching in the context of school choice eliminates justified envy in the sense that there is no unmatched student-school pair $(i,s)$ where student $i$ prefers school $s$ to her assignment and she has higher priority than some other student who is assigned a seat at school $s$. Since only the preferences of students matters in the context of school choice, the stable matching that results from the DA Mechanism Pareto dominates any other matching that eliminates justified envy and is strategy-proof. However, this matching may still be Pareto-dominated.\footnote{Do\u{g}an and Ehlers (2021) characterize the priority profiles for which there exists a Pareto improvement over the DA matching that is minimally unstable among Pareto efficient matchings.}  A Pareto efficient and strategy-proof matching is obtained by the TTC mechanism. There is no mechanism that is both Pareto efficient and stable.\footnote{See e.g. Roth (1982). Che and Tercieux (2019) show that both Pareto efficiency and stability can be achieved asymptotically using DA and TTC mechanisms when agents have uncorrelated preferences.}   

Up to now, it has been assumed that all students are myopic when they decide to join or leave some school. Myopic students do not anticipate that other students may react to their decisions. However, looking forward joining school $s^{\prime}$, a farsighted student $i$ may decide to join some school $s$ to push student $j$ out of school $s$, and later on exchanging her priority at school $s$ with another student $k$ who has priority at school $s^{\prime}$, prefers $s$ to $s^{\prime}$ and is worse ranked than $j$ at $s$.

Does the TTC mechanism lead to a stable matching when students become farsighted? To address this question, we adopt the notion of farsighted stable set for school choice problems to study the matchings that are stable when students farsightedly apply to schools while schools myopically and mechanically enroll students.\footnote{See Chwe (1994), Mauleon, Vannetelbosch and Vergote (2011), Ray and Vohra (2015, 2019), Herings, Mauleon and Vannetelbosch (2019, 2020), Luo, Mauleon and Vannetelbosch (2021) for definitions of the farsighted stable set.}  A farsighted improving path for school choice problems consists of a sequence of matchings that can emerge when farsighted students form or destroy matches based on the improvement the end matching offers them relative to the current one while myopic schools always accept any student on their priority lists unless they have full capacity. In the case of full capacity, a school accepts to replace the current match by another match if each student who leaves is replaced by a newly enrolled student who has a higher priority. A set of matchings is a farsighted stable set if (Internal Stability) for any two matchings belonging to the set, there is no farsighted improving path connecting from one matching to the other one, and (External Stability) there always exists a farsighted improving path from every matching outside the set to some matching within the set.

We show that, once students are farsighted, the matching obtained from the TTC algorithm becomes stable. A singleton set consisting of the TTC matching is a farsighted stable set. In fact, we construct a farsighted improving path from any matching leading to the TTC matching. Along the farsighted improving path, students belonging to cycles sequentially act in the order of the formation of cycles in the TTC algorithm. Looking forward towards the end matching (i.e. the TTC matching), students belonging to a cycle first get a seat at the school they have priority. Second, they leave that school, and by doing so, guaranteeing a free seat at that school. Third, they join the school they match to in the TTC matching.

Thus, the matching obtained from the TTC algorithm is not only Pareto efficient and strategy-proof, it is also farsightedly stable. On the contrary, the matching obtained from the DA algorithm may not belong to any farsightedly stable set. In addition, starting from any matching, students only need to look forward (at least) three steps ahead to have incentives for engaging a move towards the matches they have in the matching obtained from the TTC algorithm. Hence, not much farsightedness is already sufficient for stabilizing the matching obtained from the TTC algorithm.

Morill (2015) and Hakimov and Kesten (2018) introduce variations of the TTC mechanism for selecting a matching that intends to be more equitable or fair by eliminating avoidable justified envy situations. Morill (2015) proposes both the First Clinch and Trade (FCT) mechanism and the Clinch and Trade (CT) mechanism, while Hakimov and Kesten (2018) develop the Equitable Top Trading Cycles (ETTC) mechanism. We show that the matchings obtained from those three variations are farsightedly stable too. That is, a singleton set consisting of the FCT matching (CT matching / ETTC matching) is a farsighted stable set. The TTC algorithm as well as its three variations lead to Pareto efficient matchings. One may be tempted to infer that any Pareto efficient matching can be stabilized once students are farsighted. However, we show that Pareto efficiency is not a sufficient condition for a matching to be farsightedly stable.\footnote{The matching obtained from the Immediate Acceptance (IA) algorithm (i.e. the Boston mechanism) may not belong to any farsighted stable set. The IA mechanism satisfies Pareto efficiency but is not strategy-proof.}

To sum up, farsightedness stabilizes the matching obtained from the TTC algorithm while destabilizes the matching obtained from the DA algorithm, and so may tip the balance in favor of TTC or one of its variations.

In addition, Abdulkadiro\u{g}lu, Che, Pathak, Roth, and Tercieux (2020) provide both theoretical and empirical results supporting the TTC mechanism over alternative mechanisms. The TTC mechanism is justified envy minimal in the class of Pareto efficient and strategy-proof mechanisms in priority-based one-to-one matching problems. Justified envy minimal means that the mechanism satisfies Pareto efficiency with the minimal amount of (myopic) instability. In priority-based many-to-one matching problems, the TTC mechanism admits less justified envy than the Serial Dictatorship mechanism in an average sense. Recently, Do\u{g}an and Ehlers (2022) show that, for any stability comparison satisfying three basic properties, the TTC mechanism is minimally unstable among Pareto efficient and strategy-proof mechanisms when schools have unit capacities. 

The paper is organized as follows. In Section \ref{S2}, we introduce priority-based school choice problems. In Section \ref{S3}, we provide a formal description of the TTC mechanism and its algorithm. In Section \ref{S4}, we introduce the notions of farsighted improving path and farsighted stable set for school choice problems, and we provide our main result. In Section \ref{S5}, we look at how much farsightedness is needed for getting our main result. In Section \ref{S6}, we consider three variations of the TTC mechanism. In Section \ref{S7}, we conclude.

\section{School choice problems}\label{S2} 

A school choice problem is a list $\langle I,S,q,P,F\rangle$ where 
\begin{itemize}
\item[(\textbf{i})] $I=\{i_1,...,i_n\}$ is the set of students,
\item[(\textbf{ii})] $S=\{s_1,...,s_m\}$ is the set of schools,
\item[(\textbf{iii})] $q=(q_{s_1},...,q_{s_m})$ is the quota vector where $q_s$ is the number of available seats at school $s$,
\item[(\textbf{iv})] $P=(P_{i_1},...,P_{i_n})$ is the preference profile where $P_i$ is the strict preference of student $i$ over the schools and her outside option,
\item[(\textbf{v})] $F=(F_{s_1},...,F_{s_m})$ is the strict priority structure of the schools over the students.
\end{itemize}

Let $i$ be a generic student and $s$ be a generic school. We write $i$ for singletons $\{i\} \subseteq I$ and $s$ for singletons $\{s\} \subseteq S$. The preference $P_i$ of student $i$ is a linear order over $S \cup i$. Student $i$ prefers school $s$ to school $s^\prime$ if $s P_i s^\prime $. School $s$ is acceptable to student $i$ if $s P_i i $. We often write $P_i = s,s^\prime,s^{\prime \prime}$ meaning that student $i$'s most preferred school is $s$, her second best is $s^\prime$, her third best is $s^{\prime \prime}$ and any other school is unacceptable for her. Let $R_i$ be the weak preference relation associated with the strict preference relation $P_i$.\footnote{Haeringer and Klijn (2009) investigate constrained school choice problems where students can only rank a fixed number of schools.} 

The priority $F_s$ of school $s$ is a linear order over $I$. That is, $F_s$ assigns ranks to students according to their priority for school $s$. The rank of student $i$ for school $s$ is denoted $F_s(i)$ and $F_s(i) < F_s(j)$ means that student $i$ has higher priority for school $s$ than student $j$. For $s \in S, i \in I$, let $\Phi(s,i) = \{j \in I \mid F_s(j) < F_s(i)\}$ be the set of students who have higher priority than student $i$ for school $s$.

An outcome of a school choice problem is a matching $\mu : I \cup S \rightarrow 2^{I} \cup S$ such that for any $i \in I$ and any $s \in S$,
\begin{itemize}
\item[(\textbf{i})]$\mu(i) \in S \cup i$,
\item[(\textbf{ii})]$\mu(s) \in 2^{I}$,
\item[(\textbf{iii})]$\mu(i)=s \Leftrightarrow i \in \mu(s)$,
\item[(\textbf{iv})]$\# \mu(s) \leq q_s$.
\end{itemize}

Condition (i) means that student $i$ is assigned a seat at school $s$ under $\mu$ if $\mu(i)=s$ and is unassigned under $\mu$ if $\mu(i)=i$. Condition (iv) requires that no school exceeds its quota under $\mu$. That is, for any $s \in S$, we have $\# \mu(s) = \# \{i \in I \mid \mu (i) = s \} \leq q_s$. The set of all matchings is denoted $\mathcal{M}$.\footnote{Throughout the paper we use the notation $\subseteq$ for weak inclusion and $\subset$ for strict inclusion. Finally, $\#$ will refer to the notion of cardinality.} For instance, 
\[\mu=\Big( \begin{array}{cccc} i_1 & i_2 & i_3 & i_4 \\ s_2 & s_1 & s_1 & i_4\end{array}\Big)\] is the matching where student $i_1$ is assigned to school $s_2$, students $i_2$ and $i_3$ are assigned to school $s_1$ and student $i_4$ is unassigned. For convenience, we often write such matching as $\mu = \{(i_1,s_2),(i_2,s_1),(i_3,s_1),(i_4,i_4)\}$.

Given a school choice problem $\langle I,S,q,P,F \rangle$, a matching $\mu$ is stable if 
\begin{itemize}
\item[(\textbf{i})]for all $i \in I$ we have $\mu(i) R_i i$ (individual rationality),
\item[(\textbf{ii})]for all $i \in I$ and all $s \in S$, if $ s P_i \mu(i)$ then $\# \{j \in I \mid \mu (j) = s \} = q_s$ (non-wastefulness),
\item[(\textbf{iii})]for all $i,j \in I$ with $\mu(j)=s$, if $\mu(j) P_i \mu(i)$ then $j \in \Phi(s,i)$ (no justified envy).
\end{itemize}

Let $\mathcal{S}(I,S,q,P,F)$ be the set of stable matchings. A matching $\mu^\prime$ Pareto dominates a matching $\mu$ if $\mu^\prime (i) R_i \mu(i)$ for all $i \in I$ and $\mu^\prime (j) P_j \mu(j)$ for some $j \in I$. A matching is Pareto efficient if it is not Pareto dominated by another matching. Let $\mathcal{E}(I,S,q,P,F)$ be the set of Pareto efficient matchings.

A mechanism systematically selects a matching for any given school choice problem $(I,S,q,P,F)$.  A mechanism is individually rational (non-wasteful / stable / Pareto efficient) if it always selects an individually rational (non-wasteful / stable / Pareto efficient) matching. A mechanism is strategy-proof if no student can ever benefit by unilaterally misrepresenting her preferences.

\section{The Top Trading Cycles algorithm}\label{S3} 

Abdulkadiro\~{g}lu and Sönmez (2003) introduce the Top Trading Cycles (TTC) mechanism for selecting a matching for each school problem. The TTC mechanism finds a matching by means of the following TTC algorithm.

\begin{itemize}
\item[Step $1$.] Set $q_s^1 = q_s$ for all $s \in S$ where $q_s^1$ is equal to the initial capacity of school $s$ at Step $1$. Each student $i \in I$ points to the school that is ranked first in $P_i$. If there is no such school, then student $i$ points to herself and she forms a self-cycle. Each school $s \in S$ points to the student that has the highest priority in $F_s$. Since the number of students and schools are finite, there is at least one cycle. A cycle is an ordered list of distinct schools and distinct students $(s^1,i^1,s^2,...,s^l,i^l)$ where $s^1$ points to $i^1$ (denoted $s^1 \mapsto i^1$), $i^1$ points to $s^2$ ($i^1 \mapsto s^2$), $s^l$ points to $i^l$ ($s^l \mapsto i^l$) and $i^l$ points to $s^1$ ($i^l \mapsto s^1$). Each school (student) can be part of at most one cycle. Every student in a cycle is assigned a seat at the school she points to and she is removed. Similarly, every student in a self-cycle is not assigned to any school and is removed. If a school $s$ is part of a cycle, then its remaining capacity $q_s^2$ is equal to $q_s^1 - 1$. If a school $s$ is not part of any cycle, then its remaining capacity $q_s^2$ remains equal to $q_s^1$. If $q_s^2 = 0$, then school $s$ is removed. Let $C_1=\{c_1^1,c_1^2,...,c_1^{L_1}\}$ be the set of cycles in Step $1$ (where $L_1 \geq 1$ is the number of cycles in Step $1$). Let $I_1$ be the set of students who are assigned to some school at Step $1$. Let $m_1^l$ be all the matches from cycle $c_1^l$ that are formed in Step $1$ of the algorithm:
\begin{equation*}
m_1^l=\left\lbrace
\begin{matrix}
\{(i,s) \mid i,s \in c_1^l \text{ and } i\mapsto s \} & \text{if} & c_1^l \neq (j) \\
\{(j,j) \} & \text{if} & c_1^l = (j)
\end{matrix}\right.
\end{equation*}
where $(j,j)$ simply means that student $j$ who is in a self-cycle ends up being definitely unassigned to any school. Let $M_1 = \cup_{l=1}^{L_1} m_1^l$ be all the matches between students and schools formed in Step $1$ of the algorithm.

\item[Step $k\geq 2$.] Notice that $q_s^k$ keeps track of how many seats are still available at the school at Step $k$ of the algorithm. Each remaining student $i \in I \setminus \cup_{l=1}^{k-1}I_l$ points to the school $s$ that is ranked first in $P_i$ such that $q_s^k \geq 1$. If there is no such school, then student $i$ points to herself and she forms a self-cycle. Each school $s \in S$ such that $q_s^k \geq 1$ points to the student $j \in I \setminus \cup_{l=1}^{k-1}I_l$ that has the highest priority in $F_s$. There is at least one cycle. Every student in a cycle is assigned a seat at the school she points to and she is removed. Similarly, every student in a self-cycle is not assigned to any school and is removed. If a school $s$ is part of a cycle, then its remaining capacity $q_s^{k+1}$ is equal to $q_s^{k} - 1$. If a school $s$ is not part of any cycle, then its remaining capacity $q_s^{k+1}$ remains equal to $q_s^{k}$. If $q_s^{k+1} = 0$, then school $s$ is removed. Let $C_k=\{c_k^1,c_k^2,...,c_k^{L_k}\}$ be the set of cycles in Step $k$ (where $L_k \geq 1$ is the number of cycles in Step $k$). Let $I_k$ be the set of students who are assigned to some school at Step $k$. 

Let $m_k^l$ be all the matches from cycle $c_k^l$ that are formed in Step $k$ of the algorithm.
\begin{equation*}
m_k^l=\left\lbrace
\begin{matrix}
\{(i,s) \mid i,s \in c_k^l \text{ and } i\mapsto s \} & \text{if} & c_k^l \neq (j) \\
\{(j,j) \} & \text{if} & c_k^l = (j)
\end{matrix}\right.
\end{equation*}
Let $M_k = \cup_{l=1}^{L_k} m_k^l$ be all the matches between students and schools formed in Step $k$ of the algorithm.

\item[End.] The algorithm stops when all students have been removed. Let $\bar{k}$ be the step at which the algorithm stops. Let $\mu^T$ denote the matching obtained from the Top Trading Cycles mechanism and it is given by $\mu^T = \cup_{k=1}^{\bar{k}} M_k$. 
\end{itemize}

Abdulkadiro\~{g}lu and Sönmez (2003) show that the TTC mechanism is Pareto efficient and strategy-proof. TTC is also individually rational and non-wasteful, but it is not stable.

In addition to TTC, two alternative mechanisms are also central to the theory of school choice and commonly adopted all over the world: the Deferred Acceptance (DA) algorithm and the Immediate Acceptance (IA) algorithm, also known as the Boston mechanism. Let $\mu^D$ denote the matching obtained from the DA mechanism and $\mu^B$ denote the matching obtained from the IA (or Boston) mechanism.  

\section{Farsighted Stable Sets for School Choice}\label{S4} 

We adopt the notion of farsighted stable set for school choice problems to study the matchings that are stable when students farsightedly apply to schools while schools myopically and mechanically enroll students. The notion of a farsighted stable set for school choice problems is adapted from the notion of a myopic-farsighted stable set that has been introduced by Herings, Mauleon and Vannetelbosch (2020) for two-sided matching problems and by Luo, Mauleon and Vannetelbosch (2021) for network formation games.\footnote{When all agents are myopic, the myopic-farsighted stable set boils down to the pairwise CP vNM set as defined in Herings, Mauleon and Vannetelbosch (2017) for two-sided matching problems. Ehlers (2007) introduces another set-valued concept based upon the concept of vNM stable sets.} 

A farsighted improving path for school choice problems is a sequence of matchings that can emerge when farsighted students form or destroy matches based on the improvement the end matching offers them relative to the current one while myopic schools form or destroy matches based on the improvement the next matching in the sequence offers them relative to the current one.

Let $\mathcal{P}(\mu(s))$ denote the power set of the set $\mu(s)$, i.e. the set of all subsets of $\mu(s)$. 

\begin{adefinition}\label{D3} 
Given a matching $\mu$, a coalition $N \subseteq I \cup S$ is said to be able to enforce a matching $\mu^{\prime}$ over $\mu$ if the following conditions hold: 
\begin{itemize}
\item[(i)]$\mu^{\prime} (s) \notin \mathcal{P}(\mu(s)) \cup \{s\}$ implies $\mu^{\prime} (s) \setminus \mu(s) \cup \{s\} \subseteq N$ and
\item[(ii)]$\mu^{\prime} (s) \in \mathcal{P}(\mu(s)) \cup \{s\}$, $\mu^{\prime} (s) \neq \mu(s)$, implies either $s$ or $\mu(s) \setminus \mu^{\prime}(s)$ or $s$ together with a non-empty subset of $\mu(s) \setminus \mu^{\prime}(s)$ should be in $N$.
\end{itemize}
\end{adefinition}

Condition (i) says that any new match in $\mu^{\prime}$ that contains different partners than in $\mu$ should be such that $s$ and the different partners of $s$ belong to $N$. Condition (ii) states that so as to leave some (or all) positions of one existing match in $\mu$ unfilled, either $s$ or the students leaving such positions or $s$ and some non-empty subset of such students should be in $N$.

\begin{adefinition}\label{D1C} 
Let $\langle I,S,q,P,F\rangle$ be a school choice problem. A farsighted improving path from a matching $\mu \in \mathcal{M}$ to a matching $\mu^{\prime } \in \mathcal{M} \setminus \{\mu \}$
is a finite sequence of distinct matchings $\mu _{0},\ldots ,\mu _{L}$ with $\mu _{0} = \mu$ and $\mu_{L} = \mu^{\prime }$ such that for every $l \in \{0,\ldots ,L-1\}$ there is a coalition $N_{l } \subseteq I \cup S$ that can enforce $\mu_{l+1}$ from $\mu_{l}$ and
\begin{itemize}
\item[(\textbf{i})]$\mu_{L}(i) R_i \mu _{l}(i) \text{ for all } i \in N_{l} \cap I \text{ and }  \mu_{L}(j) P_j \mu _{l}(j) \text{ for some } j \in N_{l} \cap I$,
\item[(\textbf{ii})]For every $s \in N_{l} \cap S$ such that $\#\mu_l(s) + \#\{i \in I \mid i \notin \mu_l(s) , i \in \mu_{l+1}(s)\} > q_s$, there is $\{i_1,\ldots,i_J\} \subseteq \{i \in I \mid i \notin \mu_l(s), i \in \mu_{l+1}(s)\}$ and $\{j_1,\ldots,j_J\} = \{i \in I \mid i \in \mu_l(s), i \notin \mu_{l+1}(s)\}$ such that 
\begin{equation*}
\begin{split}
F_s(i_1) & < F_s(j_1) \\
F_s(i_2) & < F_s(j_2) \\
& \vdots \\
F_s(i_J) & < F_s(j_J).
\end{split}
\end{equation*}
\end{itemize}
\end{adefinition}

Notice that $\mu_l(s)$ are the students who are assigned to school $s$ in $\mu_l$ and $\{i \in I \mid i \notin \mu_l(s) , i \in \mu_{l+1}(s)\}$ are the students who join school $s$ in $\mu_{l+1}$. Thus, a farsighted improving path for school choice problems consists of a sequence of matchings where along the sequence (i) students form or destroy matches based on the improvement the end matching offers them relative to the current one while (ii) schools always accept any student on their priority lists unless they have full capacity. In the case of full capacity, a school $s \in N_l \cap S$ accepts to replace the match $\mu_l$ by $\mu_{l+1}$ if each student $i \in \{j \in I \mid j \in \mu_l(s) , j \notin \mu_{l+1}(s)\}$ who leaves or is evicted from school $s$ from $\mu_l$ to $\mu_{l+1}$ is replaced by a newly enrolled student who has a higher priority.  

Let some $\mu \in \mathcal{M}$ be given. If there exists a farsighted improving path from a matching $\mu$ to a matching $\mu^{\prime}$, then we write $\mu \rightarrow \mu^{\prime}$. The set of matchings $\mu^{\prime }\in \mathcal{M}$ such that there is a farsighted improving path from $\mu$ to $\mu^{\prime }$ is denoted by $\phi(\mu)$, so $\phi(\mu)=\{\mu^{\prime} \in \mathcal{M} \mid \mu \rightarrow \mu^{\prime} \}$.

\begin{adefinition} \label{D5} Let $\langle I,S,q,P,F\rangle$ be a school choice problem. A set of matchings $V \subseteq \mathcal{M}$ is a farsighted stable set if it satisfies:
\begin{itemize}
\item[(\textbf{i})] Internal stability (IS): For every $\mu, \mu^{\prime} \in V$,
it holds that $\mu^{\prime} \notin \phi(\mu).$
\item[(\textbf{ii})] External stability (ES): For every $\mu \in \mathcal{M} \setminus V$, it holds that $\phi(\mu) \cap V \neq \emptyset. $
\end{itemize}
\end{adefinition}

Condition (i) of Definition \ref{D5} corresponds to internal stability. For any two matchings $\mu $ and $\mu^{\prime }$ in the farsighted stable set $V$ there is no farsighted improving path connecting $\mu$ to $\mu^{\prime }$. Condition (ii) of Definition \ref{D5} expresses external stability. There always exists a farsighted improving path from every matching $\mu$ outside the farsighted stable set $V$ to some matching in $V$.

When all agents are farsighted, the notion of the farsighted stable set in Definition \ref{D5} coincides with the definition of the vNM farsightedly stable set of Mauleon, Vannetelbosch and Vergote (2011).

Given a matching $\mu \in \mathcal{M} $ with student $i \in I $ matched to school $s \in S, $ so $\mu(i) = s, $ the matching $\mu^{\prime} $ that is identical to $\mu, $ except that the match between $i$ and $s$ has been destroyed by either $i$ or $s$, is denoted by $\mu - (i,s)$. Given a matching $\mu \in \mathcal{M}$ such that $i \in I$ and $s \in S$ are not matched to one another, the matching $\mu^{\prime}$ that is identical to $\mu $, except that the pair $(i,s)$ has formed at $\mu^{\prime}$ (and some $j \in \mu(s)$ becomes unassigned if $\#\mu(s)=q_s$), is denoted by $\mu + (i,s)$.

\begin{atheorem}\label{T5}
Let $\langle I,S,q,P,F\rangle$ be a school choice problem and $\mu^T$ be the matching obtained from the Top Trading Cycles mechanism. The singleton set $\{\mu^T\}$ is a farsighted stable set.
\end{atheorem}

\begin{proof}
Since $\{\mu^T\}$ is a singleton set, internal stability (IS) is satisfied. (ES) Take any matching $\mu \neq \mu^T$, we need to show that $\phi(\mu) \ni \mu^T$. We build in steps a farsighted improving path from $\mu$ to $\mu^T$.
\begin{itemize}
\item[Step 1.1.]If $m_1^1 \subseteq \mu$ and $1 \neq L_1$ then go to Step 1.2 with $\mu_{1,1}^{\prime \prime \prime} = \mu$. If $m_1^1 \subseteq \mu$ and $1 = L_1$ then go to Step 1.End with $\mu_{1,L_1}^{\prime \prime \prime} = \mu$. If $m_1^1 \nsubseteq \mu$ then $\mu_{1,1}^{\prime} = \mu - \{(i,\mu(i)) \mid (i,\mu^T(i)) \in m_1^1 \text{ and } \mu(i) \neq i\} + \{(i,s) \mid i,s \in c_1^1 \text{ and } s \mapsto i\} - \{(j,s) \in \mu \mid s \in c_1^1 \text{, } \mu(s) \cap c_1^1 = \emptyset \text{, } \# \mu(s)=q_s \text{ and } F_s(j)>F_s(l) \text{ for all } l \in \mu(s), l \neq j \}$. That is, starting from $\mu$, looking forward towards $\mu^T$, the coalition of students belonging to $c_1^1$ has incentives to deviate to $\mu_{1,1}^{\prime}$ where each student in $c_1^1$ is assigned to the school where she has the highest priority. Students belonging to $c_1^1$ obtain their best match in $\mu^T$. Schools have incentives to accept those students because either they do not have full capacity or the new student replaces the student who had the lowest priority among the students enrolled at the school. Next, students belonging to $c_1^1$ leave their school to reach $\mu_{1,1}^{\prime \prime} = \mu_{1,1}^{\prime} - \{(i,s) \mid i,s \in c_1^1 \text{ and } s \mapsto i\}$. Next, each student belonging to $c_1^1$ joins her most preferred school to reach $\mu_{1,1}^{\prime \prime \prime} = \mu_{1,1}^{\prime \prime} + \{(i,s) \mid i,s \in c_1^1 \text{ and } i \mapsto s\}$. Schools accept those students since they have (at least) one vacant position. We reach $\mu_{1,1}^{\prime \prime \prime}$ with $m_1^1 \subseteq \mu_{1,1}^{\prime \prime \prime}$ and so students belonging to $c_1^1$ are assigned to the same school as in $\mu^T$. If $1 \neq L_1$, then go to Step 1.2. Otherwise, go to Step 1.End with $\mu_{1,L_1}^{\prime \prime \prime} = \mu_{1,1}^{\prime \prime \prime}$.

\item[Step 1.$k$.]($k>1$) If $m_1^k \subseteq \mu_{1,k-1}^{\prime \prime \prime}$ and $k \neq L_1$ then go to Step 1.k+1 with $\mu_{1,k}^{\prime \prime \prime} = \mu_{1,k-1}^{\prime \prime \prime}$. If $m_1^k \subseteq \mu_{1,k-1}^{\prime \prime \prime}$ and $k = L_1$ then go to Step 1.End with $\mu_{1,L_1}^{\prime \prime \prime} = \mu_{1,k-1}^{\prime \prime \prime}$. If $m_1^k \nsubseteq \mu_{1,k-1}^{\prime \prime \prime}$ then $\mu_{1,k}^{\prime} = \mu_{1,k-1}^{\prime \prime \prime} - \{(i,\mu_{1,k-1}^{\prime \prime \prime}(i)) \mid (i,\mu^T(i)) \in m_1^k \text{ and } \mu_{1,k-1}^{\prime \prime \prime}(i) \neq i\} + \{(i,s) \mid i,s \in c_1^k \text{ and } s \mapsto i\} - \{(j,s) \in \mu_{1,k-1}^{\prime \prime \prime} \mid s \in c_1^k \text{, } \mu_{1,k-1}^{\prime \prime \prime}(s) \cap c_1^k = \emptyset \text{, } \# \mu_{1,k-1}^{\prime \prime \prime}(s)=q_s \text{ and } F_s(j)>F_s(l) \text{ for all } l \in \mu_{1,k-1}^{\prime \prime \prime}(s), l \neq j \}$. From $\mu_{1,k-1}^{\prime \prime \prime}$, looking forward towards $\mu^T$, the coalition of students belonging to $c_1^k$ has incentives to deviate to $\mu_{1,k}^{\prime}$ where each student in $c_1^k$ is assigned to the school where she has the highest priority. Students belonging to $c_1^k$ obtain their best match in $\mu^T$. Schools have incentives to accept those students because either they do not have full capacity or the new student replaces the student who had the lowest priority among the students enrolled at the school. Next, students belonging to $c_1^k$ leave their school to reach $\mu_{1,k}^{\prime \prime} = \mu_{1,k}^{\prime} - \{(i,s) \mid i,s \in c_1^k \text{ and } s \mapsto i\}$. Next, each student belonging to $c_1^k$ joins her most preferred school to reach $\mu_{1,k}^{\prime \prime \prime} = \mu_{1,k}^{\prime \prime} + \{(i,s) \mid i,s \in c_1^k \text{ and } i \mapsto s\}$. Schools accept those students since they have (at least) one vacant position. We reach $\mu_{1,k}^{\prime \prime \prime}$ with $m_1^k \subseteq \mu_{1,k}^{\prime \prime \prime}$ and so students belonging to $c_1^k$ are assigned to the same school as in $\mu^T$. If $k \neq L_1$, then go to Step 1.$k+1$. Otherwise, go to Step 1.End with $\mu_{1,L_1}^{\prime \prime \prime} = \mu_{1,k}^{\prime \prime \prime}$.

\item[Step 1.End.] We have reached $\mu_{1,L_1}^{\prime \prime \prime}$ with $\cup_{l=1}^{L_1} m_1^{l} = M_1 \subseteq \mu_{1,L_1}^{\prime \prime \prime}$. If $\mu_{1,L_1}^{\prime \prime \prime} = \mu^T$ then the process ends. Otherwise, go to Step 2.1.

\item[Step 2.1.]If $m_2^1 \subseteq \mu_{1,L_1}^{\prime \prime \prime}$ and $1 \neq L_2$ then go to Step 2.2 with $\mu_{2,1}^{\prime \prime \prime} = \mu_{1,L_1}^{\prime \prime \prime}$. If $m_2^1 \subseteq \mu_{1,L_1}^{\prime \prime \prime}$ and $1 = L_2$ then go to Step 2.End with $\mu_{2,L_2}^{\prime \prime \prime} = \mu_{1,L_1}^{\prime \prime \prime}$. If $m_2^1 \nsubseteq \mu_{1,L_1}^{\prime \prime \prime}$ then $\mu_{2,1}^{\prime} = \mu_{1,L_1}^{\prime \prime \prime} - \{(i,\mu_{1,L_1}^{\prime \prime \prime}(i)) \mid (i,\mu^T(i)) \in m_2^1 \text{ and } \mu_{1,L_1}^{\prime \prime \prime}(i) \neq i\} + \{(i,s) \mid i,s \in c_2^1 \text{ and } s \mapsto i\} - \{(j,s) \in \mu_{1,L_1}^{\prime \prime \prime} \mid s \in c_2^1 \text{, } \mu_{1,L_1}^{\prime \prime \prime}(s) \cap c_2^1 = \emptyset \text{, } \# \mu_{1,L_1}^{\prime \prime \prime}(s)=q_s \text{ and } F_s(j)>F_s(l) \text{ for all } l \in \mu_{1,L_1}^{\prime \prime \prime}(s), l \neq j \}$. Starting from $\mu_{1,L_1}^{\prime \prime \prime}$, looking forward towards $\mu^T$, the coalition of students belonging to $c_2^1$ has now incentives to deviate to $\mu_{2,1}^{\prime}$ where each student in $c_2^1$ is assigned to the school where she has the highest priority among students belonging to $I \setminus I_1$. Remember that $I_1$ is the set of students who are involved in $M_1$. Given $M_1 \subseteq \mu_{1,L_1}^{\prime \prime \prime}$ remains fixed, students belonging to $c_2^1$ obtain their best match in $\mu^T$. Schools have incentives to accept those students because either they do not have full capacity or the new student replaces the student who had the lowest priority among the students enrolled at the school. Next, students belonging to $c_2^1$ leave their school to reach $\mu_{2,1}^{\prime \prime} = \mu_{2,1}^{\prime} - \{(i,s) \mid i,s \in c_2^1 \text{ and } s \mapsto i\}$. Next, each student belonging to $c_2^1$ joins her most preferred school (constrained to $M_1$ being fixed) to reach $\mu_{2,1}^{\prime \prime \prime} = \mu_{2,1}^{\prime \prime} + \{(i,s) \mid i,s \in c_2^1 \text{ and } i \mapsto s\}$. Schools accept those students since they have (at least) one vacant position. We reach $\mu_{2,1}^{\prime \prime \prime}$ with $m_2^1 \subseteq \mu_{2,1}^{\prime \prime \prime}$ and so students belonging to $c_2^1$ are assigned to the same school as in $\mu^T$. If $1 \neq L_2$, then go to Step 2.2. Otherwise, go to Step 2.End with $\mu_{1,L_2}^{\prime \prime \prime} = \mu_{2,1}^{\prime \prime \prime}$.

\item[Step 2.$k$.]($k>1$) If $m_2^k \subseteq \mu_{2,k-1}^{\prime \prime \prime}$ and $k \neq L_2$ then go to Step 2.k+1 with $\mu_{2,k}^{\prime \prime \prime} = \mu_{2,k-1}^{\prime \prime \prime}$. If $m_2^k \subseteq \mu_{2,k-1}^{\prime \prime \prime}$ and $k = L_2$ then go to Step 2.End with $\mu_{2,L_2}^{\prime \prime \prime} = \mu_{2,k-1}^{\prime \prime \prime}$. If $m_2^k \nsubseteq \mu_{2,k-1}^{\prime \prime \prime}$ then $\mu_{2,k}^{\prime} = \mu_{2,k-1}^{\prime \prime \prime} - \{(i,\mu_{2,k-1}^{\prime \prime \prime}(i)) \mid (i,\mu^T(i)) \in m_2^k \text{ and } \mu_{2,k-1}^{\prime \prime \prime}(i) \neq i\} + \{(i,s) \mid i,s \in c_2^k \text{ and } s \mapsto i\} - \{(j,s) \in \mu_{2,k-1}^{\prime \prime \prime} \mid s \in c_2^k \text{, } \mu_{2,k-1}^{\prime \prime \prime}(s) \cap c_2^k = \emptyset \text{, } \# \mu_{2,k-1}^{\prime \prime \prime}(s)=q_s \text{ and } F_s(j)>F_s(l) \text{ for all } l \in \mu_{2,k-1}^{\prime \prime \prime}(s), l \neq j \}$. Starting from $\mu_{2,k-1}^{\prime \prime \prime}$, looking forward towards $\mu^T$, the coalition of students belonging to $c_2^k$ has now incentives to deviate to $\mu_{2,k}^{\prime}$ where each student in $c_2^k$ is assigned to the school where she has the highest priority among students belonging to $I \setminus I_1$. Given $M_1 \subseteq \mu_{2,k-1}^{\prime \prime \prime}$ remains fixed, students belonging to $c_2^k$ obtain their best match in $\mu^T$. Schools have incentives to accept those students because either they do not have full capacity or the new student replaces the student who had the lowest priority among the students enrolled at the school. Next, students belonging to $c_2^k$ leave their school to reach $\mu_{2,k}^{\prime \prime} = \mu_{2,k}^{\prime} - \{(i,s) \mid i,s \in c_2^k \text{ and } s \mapsto i\}$. Next, each student belonging to $c_2^k$ joins her most preferred school (constrained to $M_1$ being fixed) to reach $\mu_{2,k}^{\prime \prime \prime} = \mu_{2,k}^{\prime \prime} + \{(i,s) \mid i,s \in c_2^k \text{ and } i \mapsto s\}$. Schools accept those students since they have (at least) one vacant position. We reach $\mu_{2,k}^{\prime \prime \prime}$ with $m_2^k \subseteq \mu_{2,k}^{\prime \prime \prime}$ and so students belonging to $c_2^k$ are assigned to the same school as in $\mu^T$. If $k \neq L_2$, then go to Step 2.$k+1$. Otherwise, go to Step 2.End with $\mu_{2,L_2}^{\prime \prime \prime} = \mu_{2,k}^{\prime \prime \prime}$.

\item[Step 2.End.] We have reached $\mu_{2,L_2}^{\prime \prime \prime}$ with $M_1 \cup M_2 \subseteq \mu_{2,L_2}^{\prime \prime \prime}$. If $\mu_{2,L_2}^{\prime \prime \prime} = \mu^T$ then the process ends. Otherwise, go to Step 3.1.

\item[End.]The process goes on until we reach $\mu_{\bar{k},L_{\bar{k}}}^{\prime \prime \prime} = \cup_{k=1}^{\bar{k}} M_k = \mu^T$.

\end{itemize}
\end{proof}

The matching obtained from the TTC algorithm is always Pareto efficient but may not be stable when students are myopic. Theorem \ref{T5} shows that, once students are farsighted, the matching obtained from the TTC algorithm becomes stable.\footnote{This result is robust to the incorporation of various forms of maximality in the definition of farsighted improving path, like the strong rational expectations farsighted stable set in Dutta and Vohra (2017) and absolute maximality as in Ray and Vohra (2019). See also Herings, Mauleon and Vannetelbosch (2020).}  Example \ref{Ex1} highlights Theorem \ref{T5}. In addition, it shows that, once students are farsighted, the matching obtained from the Deferred Acceptance (DA) algorithm may become unstable. 
 
\begin{aexample}[Haeringer, 2017]\label{Ex1} 
Consider a school choice problem $\langle I,S,q,P,F\rangle$ with $I=\{i_1,i_2,i_3,i_4\}$ and $S=\{s_1,s_2,s_3\}$. Students' preferences and schools' priorities and capacities are as follows.
\begin{center}
\begin{tabular}{cccc}
\multicolumn{4}{c}{Students}\\
\hline
$P_{i_1}$& $P_{i_2}$ & $P_{i_3}$ & $P_{i_4}$ \\
\hline
$s_1$ & $s_1$ & $s_2$ & $s_1$  \\
$s_2$ & $s_2$ & $s_1$  & $s_3$ \\
$s_3$ & $s_3$ & $s_3$  & $s_2$ 
\end{tabular}
~\qquad~
\begin{tabular}{lccc}
\multicolumn{4}{c}{Schools}\\
\hline
&     $F_{s_1}$& $F_{s_2}$ & $F_{s_3}$\\
\hline
$q_s$ & 2 & 1 & 1\\
\hline
&$i_1$ & $i_1$ & $i_2$\\
&$i_3$ & $i_2$ & $i_3$\\
&$i_4$ & $i_4$ & $i_4$\\
&$i_2$ & $i_3$ & $i_1$
    \end{tabular}
  \end{center}

\end{aexample}

Using Example \ref{Ex1} we provide the basic intuition behind Theorem \ref{T5} and its proof. In Example \ref{Ex1}, $\mu^T=\{(i_1,s_1),(i_2,s_1),(i_3,s_2),(i_4,s_3)\}$ is the matching obtained from the TTC algorithm. In the first round of the TTC algorithm, there is one cycle where student $i_1$ points to school $s_1$ and school $s_1$ points to student $i_1$. That is, $C_1=\{c_1^1\}$ with $c_1^1=\{s_1,i_1\}$. Student $i_1$ is matched to school $s_1$: $m_1^1=\{(i_1,s_1)\}$ and school $s_1$ has only one leftover seat. In the second round of the TTC algorithm, there is one cycle where student $i_2$ points to school $s_1$, school $s_1$ points to student $i_3$, student $i_3$ points to school $s_2$ and school $s_2$ points to student $i_2$. That is, $C_2=\{c_2^1\}$ with $c_2^1=\{s_1,i_3,s_2,i_2\}$. Student $i_2$ is matched to school $s_1$ and student $i_3$ is matched to school $s_2$: $m_2^1=\{(i_2,s_1),(i_3,s_2)\}$, and so $i_2$ and $i_3$ exchange their priority. In the third round of the TTC algorithm, there is only one leftover student, $i_4$, who points to school $s_3$ and school $s_3$ points to student $i_4$. That is, $C_3=\{c_3^1\}$ with $c_3^1=\{s_3,i_4\}$. Student $i_4$ is matched to school $s_3$: $m_3^1=\{(i_4,s_3)\}$, and so $\mu^T=m_1^1 \cup m_2^1 \cup m_3^1$. 

From Theorem \ref{T5} we know that $\{\mu^T\}$ is a farsighted stable set. Indeed, from any $\mu \neq \mu^T$ there exists a farsighted improving path leading to $\mu^T$. Take for instance the matching $\mu_0 = \{(i_1,s_1),(i_2,s_2),(i_3,s_3),(i_4,s_1)\}$. We now construct a farsighted improving from $\mu_0$  to $\mu^T = \{(i_1,s_1),(i_2,s_1),(i_3,s_2),(i_4,s_3)\} = \mu_4$ following the steps as in the proof of Theorem \ref{T5}. First, we consider students and schools belonging to the cycles in $C_1$. Since $m_1^1 =\{(i_1,s_1)\} \subseteq \mu_0$, student $i_1$ stays matched to school $s_1$ along the farsighted improving path, i.e. $m_1^1=\{(i_1,s_1)\} \subseteq \mu_l$, $0 \leq l \leq 4$. Next, we consider students and schools belonging to the cycles in $C_2$. Notice that $m_2^1 =\{(i_2,s_1),(i_3,s_2)\} \cap  \mu_0 = \emptyset$. Looking forward towards $\mu^T$,  the coalition $N_0=\{i_2,i_3,s_1,s_2\}$ deviates so that student $i_3$ joins school $s_1$ and student $i_2$ joins schools $s_2$ to reach the matching $\mu_1 = \{(i_1,s_1),(i_2,s_2),(i_3,s_1),(i_4,i_4)\}$ where students $i_2$ and $i_3$ are matched to the schools where they have priority. By doing so, they push student $i_4$ out of school $s_1$. Next, the coalition $N_1=\{i_2,i_3\}$ deviates so that students $i_2$ and $i_3$ leave, respectively, schools $s_2$ and $s_1$ to reach the matching $\mu_2 = \{(i_1,s_1),(i_2,i_2),(i_3,i_3),(i_4,i_4)\}$ where both students are not assigned to any school. They are temporarily worse off, but they anticipate to end up in $\mu^T$. Next, the coalition $N_2=\{i_2,i_3,s_1,s_2\}$ deviates so that student $i_2$ joins school $s_1$ and student $i_3$ joins schools $s_2$ to reach the matching $\mu_3 = \{(i_1,s_1),(i_2,s_1),(i_3,s_2),(i_4,i_4)\}$ with $m_2^1 = \{(i_2,s_1),(i_3,s_2)\} \subseteq \mu_3$. Both schools accept to enroll those students because they are not at full capacity. Finally, we consider students and schools belonging to the cycles in $C_3$. Since $m_3^1 =\{(i_4,s_3)\} \cap  \mu_3 = \emptyset$, the coalition $N_3=\{i_4,s_3\}$ deviates so that student $i_4$ joins school $s_3$ to form the match $(i_4,s_3)$ and to reach the matching $\mu_4 = \mu^T$. Thus, $\mu^T \in \phi(\mu_0)$.

In Example \ref{Ex1}, $\mu^D=\{(i_1,s_1),(i_2,s_2),(i_3,s_1),(i_4,s_3)\}$ is the matching obtained from the Deferred Acceptance (DA) algorithm, $\mu^{B}=\{(i_1,s_1),(i_2,s_3),(i_3,s_2),(i_4,s_1)\}$ is the matching obtained from the Immediate Acceptance (IA) algorithm (i.e. the Boston mechanism). Thus, $\mu^T \neq \mu^D \neq \mu^B$.

Since students are at least as well off and some of them ($i_2$ and $i_3$) are strictly better off in $\mu^T$ than in $\mu^D$, we have that there is no farsighted improving path from $\mu^T$ to $\mu^D$. That is, $\mu^D \notin \phi(\mu^T)$. Hence, $\{\mu^D\}$ is not a farsighted stable set since (ES) is violated. Let
\begin{equation*}
\begin{split}
\mu^1 & = \{(i_1,s_1),(i_2,i_2),(i_3,s_2),(i_4,s_1)\}, \\
\mu^2 & = \{(i_1,s_1),(i_2,s_3),(i_3,s_2),(i_4,s_1)\} = \mu^B, \\
\mu^3 & = \{(i_1,s_1),(i_2,s_2),(i_3,i_3),(i_4,s_1)\}, \\
\mu^4 & = \{(i_1,s_1),(i_2,s_2),(i_3,s_3),(i_4,s_1)\}, \\
\mu^5 & = \{(i_1,s_1),(i_2,s_2),(i_3,s_1),(i_4,i_4)\}.
\end{split}
\end{equation*}

\noindent Computing the farsighted improving paths emanating from $\mu^T$, we get $\phi(\mu^T)=\{\mu^1,\mu^2,\mu^3,\mu^4\}$. Notice that $\mu^5 \notin \phi(\mu^T)$ since student $i_4$ is worst off in $\mu^5$ than in $\mu^T$. From $\mu^1$, $\mu^2$, $\mu^3$, $\mu^4$ and $\mu^5$, there is a farsighted improving to $\mu^D$. That is, $ \mu^D \in \phi(\mu)$ for $\mu \in \{\mu^1,\mu^2,\mu^3,\mu^4,\mu^5\}$. From $\mu^D$ there is only a farsighted improving path to $\mu^T$; i.e. $\phi(\mu^D)=\{\mu^T\}$. For a set $V \supseteq \{\mu^D\}$ to be a farsighted stable set, we need that (i) $\mu^T \notin V$ (otherwise (IS) is violated), (ii) a single $\mu \in \{\mu^1,\mu^2,\mu^3,\mu^4\}$ should belong to $V$ to satisfy (ES) since $\mu^D \notin \phi(\mu^T)$. But, $V$ would then violate (IS) since $\mu^D \in \phi(\mu)$ for $\mu \in \{\mu^1,\mu^2,\mu^3,\mu^4,\mu^5\}$. Thus, there is no $V$ such that $\mu^D \in V$ that is a farsighted stable set in Example \ref{Ex1}. 

Since $\phi(\mu^D)=\{\mu^T\}$, there is no farsighted improving path from $\mu^D$ to $\mu^B$. Thus, $V = \{\mu^B\}$ does not satisfy (ES), and hence $V = \{\mu^B\}$ is not a farsighted stable set. Moreover, a set $V \supseteq \{\mu^B,\mu^D\}$ cannot be a farsighted stable since $\mu^D \in \phi(\mu^B)$. Otherwise, $V$ would violate (IS) since there is a farsighted stable improving path from $\mu^B$ to $\mu^D$.

Is $V=\{\mu^T\}$ the unique farsighted stable set in Example \ref{Ex1}? The answer is yes because any $V$ such that $\{\mu^D,\mu^T\}\nsubseteq V$ violates (ES). Thus, Example \ref{Ex1} provides an example where both the matching obtained from the Deferred Acceptance (DA) algorithm and the matching obtained from the Immediate Acceptance (IA) algorithm are not stable once students are farsighted.

\begin{aremark}
There are school choice problems such that the matching obtained from the Deferred Acceptance (DA) algorithm does not belong to any farsightedly stable set.
\end{aremark}

Since the matching obtained from the Immediate Acceptance (IA) algorithm is Pareto efficient, Example \ref{Ex1} also shows that there are school choice problems where some Pareto efficient matching does not belong to any farsighted stable set. Thus, Pareto efficiency is not a sufficient condition for guaranteeing the stability of a matching when students are farsighted.

\begin{aremark}
There are school choice problems such that some Pareto efficient matching does not belong to any farsighted stable set.
\end{aremark}

\begin{acorollary}\label{CCOAL}
Let $\langle I,S,q,P,F\rangle$ be a school choice problem and $\mu^T$ be the matching obtained from the Top Trading Cycles mechanism. From any $\mu \neq \mu^T$ there is a farsighted improving path to $\mu^T$ with $\mu _{0} = \mu$ and $\mu_{L} = \mu^T$ such that for every $l \in \{0,\ldots ,L-1\}$ there is a coalition $N_{l } \subseteq \bigcup_{k=1}^{\bar{k}} C_k $ that enforces $\mu_{l+1}$ from $\mu_{l}$.
\end{acorollary}

Corollary \ref{CCOAL} follows from the proof of Theorem \ref{T5}. Notice that $\bigcup_{k=1}^{\bar{k}} C_k$ is simply the collection of sets where each element is a set consisting of students and schools belonging to a cycle obtained from the TTC algorithm. Definition \ref{D1C} of a farsighted improving path is quite permissive in terms of the size of the coalition $N_l$ that enforces $\mu_{l+1}$ from $\mu_{l}$. However, Corollary \ref{CCOAL} tells us that there exists a farsighted improving path from $\mu \neq \mu^T$ to $\mu^T$ with $\mu _{0} = \mu$ and $\mu_{L} = \mu^T$ such that for every $l \in \{0,\ldots ,L-1\}$ the coalition $N_{l }$ that enforces $\mu_{l+1}$ from $\mu_{l}$ consists of students (and possibly schools) who are part of the same cycle in the TTC algorithm. Thus, for getting Theorem \ref{T5}, it is sufficient to allow a deviating coalition (involving more than one student) to be composed of exclusively students (and possibly their schools) who are exchanging their priorities among themselves in the TTC algorithm. Such restriction seems not too demanding since students who coordinate their moves are the ones who exchange their priorities.

\section{Limited Farsightedness}\label{S5} 

How much farsightedness from the students do we need to stabilize the matching obtained from the TTC algorithm? To answer this question we propose the notion of horizon-$k$ farsighted stable set for school choice problems to study the matchings that are stable when students are limited in their degree of farsightedness. A horizon-$k$ farsighted improving path for school choice problems is a sequence of matchings that can emerge when limited farsighted students form or destroy matches based on the improvement the $k$-steps ahead matching offers them relative to the current one while myopic schools form or destroy matches based on the improvement the next matching in the sequence offers them relative to the current one. A set of matchings is a horizon-$k$ farsighted stable set if (IS) for any two matchings belonging to the set, there is no horizon-$k$ farsighted improving path connecting from one matching to the other one, and (ES) there always exists a horizon-$k$ farsighted improving path from every matching outside the set to some matching within the set.

\begin{adefinition}\label{DMFLI} 
Let $\langle I,S,q,P,F\rangle$ be a school choice problem. A horizon-$k$ farsighted improving path from a matching $\mu \in \mathcal{M}$ to a matching $\mu^{\prime } \in \mathcal{M} \setminus \{\mu \}$ is a finite sequence of distinct matchings $\mu _{0},\ldots ,\mu _{L}$ with $\mu _{0} = \mu$ and $\mu_{L} = \mu^{\prime }$ such that for every $l \in \{0,\ldots ,L-1\}$ there is a coalition $N_{l } \subseteq I \cup S$ that can enforce $\mu_{l+1}$ from $\mu_{l}$ and
\begin{itemize}
\item[(\textbf{i})]$\mu_{\min \{l+k,L\}}(i) R_i \mu _{l}(i) \text{ for all } i \in N_{l} \cap I \text{ and }  \mu_{\min \{l+k,L\}}(j) P_j \mu _{l}(j) \text{ for some } j \in N_{l} \cap I$,
\item[(\textbf{ii})]For every $s \in N_{l} \cap S$ such that $\#\mu_l(s) + \#\{i \in I \mid i \notin \mu_l(s) , i \in \mu_{l+1}(s)\} > q_s$, there is $\{i_1,\ldots,i_J\} \subseteq \{i \in I \mid i \notin \mu_l(s), i \in \mu_{l+1}(s)\}$ and $\{j_1,\ldots,j_J\} = \{i \in I \mid i \in \mu_l(s), i \notin \mu_{l+1}(s)\}$ such that 
\begin{equation*}
\begin{split}
F_s(i_1) & < F_s(j_1) \\
F_s(i_2) & < F_s(j_2) \\
& \vdots \\
F_s(i_J) & < F_s(j_J).
\end{split}
\end{equation*}
\end{itemize}
\end{adefinition}

Definition \ref{DMFLI} tells us that a horizon-$k$ farsighted improving path for school choice problems consists of a sequence of matchings where along the sequence students form or destroy matches based on the improvement the $k$-steps ahead matching offers them relative to the current one. Precisely, along a horizon-$k$ farsighted improving path, each time some student $i$ is on the move she is comparing her current match (i.e. $\mu_l(i)$) with the match she will get $k$-steps ahead on the sequence (i.e. $\mu_{l+k}(i)$) except if the end matching of the sequence lies within her horizon (i.e. $L < l+k$). In such a case, she simply compares her current match (i.e. $\mu_l(i)$) with the end match (i.e. $\mu_L$). Schools continue to accept any student on their priority lists unless they have full capacity. In the case of full capacity, a school $s \in N_l \cap S$ accepts to replace the match $\mu_l$ by $\mu_{l+1}$ if each student $i \in \{j \in I \mid j \in \mu_l(s) , j \notin \mu_{l+1}(s)\}$ who leaves or is evicted from school $s$ from $\mu_l$ to $\mu_{l+1}$ is replaced by a newly enrolled student who has a higher priority. 

Let some $\mu \in \mathcal{M}$ be given. If there exists a horizon-$k$ farsighted improving path from a matching $\mu$ to a matching $\mu^{\prime}$, then we write $\mu \rightarrow_k \mu^{\prime}$. The set of matchings $\mu^{\prime }\in \mathcal{M}$ such that there is a horizon-$k$ farsighted improving path from $\mu$ to $\mu^{\prime }$ is denoted by $\phi_k(\mu)$, so $\phi_k(\mu)=\{\mu^{\prime} \in \mathcal{M} \mid \mu \rightarrow_k \mu^{\prime} \}$.

\begin{adefinition} \label{D6} Let $\langle I,S,q,P,F\rangle$ be a school choice problem. A set of matchings $V \subseteq \mathcal{M}$ is a horizon-$k$ farsighted stable set if it satisfies:
\begin{itemize}
\item[(\textbf{i})] Internal stability (IS): For every $\mu, \mu^{\prime} \in V$,
it holds that $\mu^{\prime} \notin \phi_k(\mu).$
\item[(\textbf{ii})] External stability (ES): For every $\mu \in \mathcal{M} \setminus V$, it holds that $\phi_k(\mu) \cap V \neq \emptyset. $
\end{itemize}
\end{adefinition}

From the construction of a farsighted improving path in the proof of Theorem \ref{T5} we have that students belonging to a cycle only need to look forward three steps ahead to have incentives for engaging a move towards the matches they have in the matching obtained from the TTC algorithm, $\mu^T$. Once they reach those matches they do not move afterwards. The three steps consist of (i) getting first a seat at the school they have priority, (ii) leaving that school and by doing so, guaranteeing a free seat at that school, (iii) joining the school they match to in $\mu^T$. Hence, for $k \geq 3$, there exists a horizon-$k$ farsighted improving from any $\mu \neq \mu^T$ to $\mu^T$, and so $\{\mu^T\}$ is a horizon-$k$ farsighted stable set.\footnote{In Example \ref{Ex1}, it is sufficient for the students who belong to $c_2^1$ to look forward towards $\mu_3$ when they participate to the moves from $\mu_0$ to $\mu_4$. Indeed, they are not affected by the move from $\mu_3$ to $\mu_4$ since they remain with the same matches.} 

\begin{acorollary}\label{CORLIM}
Let $\langle I,S,q,P,F\rangle$ be a school choice problem and $\mu^T$ be the matching obtained from the Top Trading Cycles mechanism. The singleton set $\{\mu^T\}$ is a horizon-$k$ farsighted stable set for $k \geq 3$.
\end{acorollary}

\section{Three Variations of The TTC Algorithm}\label{S6} 

\subsection{First Clinch and Trade Algorithm}

Morill (2015) introduces two variations of the Top Trading Cycles mechanism for selecting a matching for each school problem: the First Clinch and Trade mechanism (FCT) and the Clinch and Trade mechanism (CT). Both mechanisms intend to mitigate the following problem. In the TTC mechanism, if a student $i$'s most preferred school is $s$ and the student has one of the $q_s$ highest priorities at $s$, then $i$ is always assigned to $s$. However, until $i$ has the highest priority at $s$, the TTC mechanism allows $i$ to trade her priority at other schools to be assigned to $s$. Such trade may cause distortions regarding the elimination of justified envy. 

In the First Clinch and Trade algorithm (FCT), a student that initially has one of the $q_s$ highest priorities at a school $s$ (she is guaranteed a seat at $s$), cannot trade with another student to get $s$. The FCT algorithm runs basically the TTC algorithm but, at each round, if a student points at a school where she is guaranteed a seat, the student is assigned to the school and cannot trade her priority. For the remaining students, the TTC is run and the students who have the highest priorities at some schools are allowed to trade their priorities and are assigned their top choices.

\begin{aexample}[Morrill, 2015]\label{Ex2} 
Consider a school choice problem $\langle I,S,q,P,F\rangle$ with $I=\{i_1,i_2,i_3\}$ and $S=\{s_1,s_2\}$. Students' preferences and schools' priorities and capacities are as follows.
\begin{center}
\begin{tabular}{cccc}
\multicolumn{3}{c}{Students}\\
\hline
$P_{i_1}$& $P_{i_2}$ & $P_{i_3}$  \\
\hline
$s_2$ & $s_1$ & $s_2$   \\
$s_1$ & $s_2$ & $s_1$   
\end{tabular}
~\qquad~
\begin{tabular}{lcc}
\multicolumn{3}{c}{Schools}\\
\hline
&    $F_{s_1}$& $F_{s_2}$ \\
\hline
$q_s$ & 2 & 1 \\
\hline
&$i_1$ & $i_2$ \\
&$i_2$ & $i_3$ \\
&$i_3$ & $i_1$ 
    \end{tabular}
  \end{center}

\end{aexample}

Let $\mu^F$ be the matching obtained from the FCT mechanism. A formal description of the FCT algorithm can be found in Appendix \ref{AP1}. By means of Example \ref{Ex2} we illustrate the mechanism behind the FCT algorithm. In the first round, each student points to her top choice school. That is, $i_1$ points to $s_2$, $i_2$ points to $s_1$ and $i_3$ points to $s_2$. Student $i_1$ and student $i_2$ are guaranteed admissions to school $s_1$ since both have one of the two highest rankings at school $s_1$. Student $i_2$ is also guaranteed admission to school $s_2$ since she is ranked first at school $s_2$. Student $i_2$ is pointing to $s_1$, and so she is clinched to school $s_1$ and the match $(i_2,s_1)$ is formed. Student $i_1$ is not pointing to a school where she is guaranteed admission. Hence, she is not clinched to any school and she participates next with $i_3$ to the trading procedure. Schools $s_1$ and $s_2$ point to their highest ranked student, respectively $i_1$ and $i_2$. Hence, there is no cycle and no match is formed. In the second round, each student points to her top choice school that has still available capacity. That is, $i_1$ points to $s_2$ and $i_3$ points to $s_2$. Guaranteed admissions do not change. Hence, nor $i_1$ nor $i_3$ are clinched to some school and so they participate next to the trading procedure. School $s_1$ points to $i_1$ while school $s_2$ points now to $i_3$ and so the match $(i_3,s_2)$ is formed. Student $i_1$ remains unmatched. In the third round, each remaining student points to her preferred school that has still available capacity. That is, student $i_1$ points now to school $s_1$. Since she is guaranteed admission to school $s_1$, she is clinched and assigned to school $s_1$. We obtain the matching $\mu^F=\{(i_1,s_1),(i_2,s_1),(i_3,s_2)\}$. 

The matching obtained from the FCT algorithm differs from the matching obtained from the TTC algorithm, $\mu^{T}=\{(i_1,s_2),(i_2,s_1),(i_3,s_1)\}$. In the TTC mechanism, students $i_1$ and $i_2$ first exchange their priorities to form the matches $(i_1,s_2)$ and $(i_2,s_1)$. Student $i_1$ has priority at $s_1$ while student $i_2$ has priority at $s_2$. However, student $i_3$ is ranked above student $i_1$ at school $s_2$. In addition, student $i_2$ is guaranteed admission at school $s_1$. The FCT mechanism intends to remedy to such drawback.

Morrill (2015) shows that the First Clinch and Trade mechanism (FCT) is Pareto efficient, strategy-proof, non-bossy, group strategy-proof, reallocation proof and independent of the order in which cycles are processed.

Take for instance the matching $\mu_0 = \{(i_1,s_2),(i_2,s_1),(i_3,s_1)\} = \mu^T$. We now construct a farsighted improving from $\mu_0$  to $\mu^F = \{(i_1,s_1),(i_2,s_1),(i_3,s_2)\} = \mu_2$ following the steps as in the proof of Theorem \ref{THFCT}. First, we consider student $i_2$ who is the only student to clinched in the first round of the FCT. Since student $i_2$ is matched to school $s_1$ in both $\mu_0$ and $\mu_F$ she does not participate to any deviation and remains clinched to $s_1$ along the farsighted improving path. That is, $(i_2,s_1) \in \mu_l$, $l=0,1,2$. There is no cycle between schools and students who are not clinched in the first round of the FCT. In the second round of the FCT, none of the remaining students is clinched to some school. However, student $i_3$ and school $s_2$ form a cycle. So, looking forward towards $\mu^T$, the coalition $N_0=\{i_3,s_2\}$ deviates from $\mu_0$ so that student $i_3$ joins school $s_2$ to reach the matching $\mu_1 = \{(i_1,i_1),(i_2,s_1),(i_3,s_2)\}$. Student $i_3$ is matched to her preferred school $s_2$ in $\mu_1$ and she has a higher priority than $i_1$ at $s_2$. By doing so, student $i_1$ is pushed out of school $s_2$. In the third round of the FCT, student $i_1$ points to school $s_1$ and is guaranteed admission at school $s_1$. So, from $\mu_1$, the coalition $N_1=\{i_1,s_1\}$ deviates so that student $i_1$ joins school $s_1$ to reach the matching $\mu_2 = \{(i_1,s_1),(i_2,s_1),(i_3,s_2)\} = \mu^F$. Thus, $\mu^F \in \phi(\mu^T)$.

In fact, it holds in general that, from any $\mu \neq \mu^F$ there exists a farsighted improving path leading to $\mu^F$. So, $\{\mu^F\}$ is a farsighted stable set and the matching obtained from the FCT algorithm preserves the property of being stable once students are farsighted.

\begin{atheorem}\label{THFCT}
Let $\langle I,S,q,P,F\rangle$ be a school choice problem and $\mu^F$ be the matching obtained from the First Clinch and Trade mechanism. The singleton set $\{\mu^F\}$ is a farsighted stable set.
\end{atheorem}

The proof of Theorem \ref{THFCT} can be found in Appendix \ref{AP1}.

\subsection{Clinch and Trade Algorithm}

While the FCT algorithm does not update the students that are able to clinch her most preferred school, the Clinch and Trade algorithm (CT) removes from the priority list of each school the students that are guaranteed a seat. As a result, the priorities of the remaining students weakly improve and thus, students that initially are not guaranteed a seat at their most preferred school may now be guaranteed one of the remaining seats. 

\begin{aexample}[Morrill, 2015]\label{Ex3} 
Consider a school choice problem $\langle I,S,q,P,F\rangle$ with $I=\{i_1,i_2,i_3,i_4\}$ and $S=\{s_1,s_2,s_3\}$. Students' preferences and schools' priorities and capacities are as follows.
\begin{center}
\begin{tabular}{cccc}
\multicolumn{4}{c}{Students}\\
\hline
$P_{i_1}$& $P_{i_2}$ & $P_{i_3}$ & $P_{i_4}$ \\
\hline
$s_2$ & $s_1$ & $s_2$ & $s_3$  \\
$s_1$ & $s_2$ & $s_1$  &  

\end{tabular}
~\qquad~
\begin{tabular}{lccc}
\multicolumn{4}{c}{Schools}\\
\hline
&     $F_{s_1}$& $F_{s_2}$ & $F_{s_3}$\\
\hline
$q_s$ & 2 & 1 & 1\\
\hline
&$i_4$ & $i_2$ & $i_4$\\
&$i_1$ & $i_3$ &      \\
&$i_2$ & $i_1$ &      \\
&$i_3$ & $i_4$ & 
    \end{tabular}
  \end{center}

\end{aexample}

Let $\mu^C$ be the matching obtained from the CT mechanism. A formal description of the CT algorithm can be found in Appendix \ref{AP2}. By means of Example \ref{Ex3} we illustrate the mechanism behind the CT algorithm. In the first round of the clinching procedure, each student points to her top choice school. That is, student $i_1$ points to school $s_2$, student $i_2$ points to school $s_1$, student $i_3$ points to school $s_2$ and student $i_4$ points to school $s_3$. Student $i_4$ is guaranteed admission to both school $s_1$ and school $s_3$, student $i_2$ is guaranteed admission to school $s_2$, and student $i_1$ is guaranteed admission to school $s_1$. Hence, only student $i_4$ is assigned and clinched to school $s_3$. In the second round of the clinching procedure, each remaining student still points to her top choice school. Since student $i_4$ is clinched to school $s_3$, guaranteed admissions are updated as follows. Student $i_2$ is now guaranteed admission to both school $s_1$ and $s_2$, while student $i_1$ is still guaranteed admission to school $s_1$. Hence, only student $i_2$ is assigned and clinched to school $s_1$. In the third round of the clinching procedure, each remaining student still points to her top choice school. Since students $i_2$ and $i_4$ are, respectively, clinched to schools $s_1$ and $s_3$, guaranteed admissions are updated as follows. Student $i_1$ is still guaranteed admission to school $s_1$, while student $i_3$ is now guaranteed admission to school $s_2$. Hence, only student $i_3$ is assigned and clinched to school $s_2$. Since student $i_1$ remains pointing at $s_2$, the iterated clinching procedure stops and leads to the matches $(i_2,s_1)$, $(i_3,s_2)$ and $(i_4,s_3)$. Next, at most one round of the trading procedure takes place before proceeding (if necessary) again with the iterated clinching procedure. Each remaining student points to her top choice school that has still available capacity. Each school with available capacity points to the remaining student who has the highest priority. That is, $i_1$ points to $s_1$ and $s_1$ points to $i_1$ to form a cycle. So, $i_1$ is assigned to $s_1$ and we obtain $\mu^C= \{(i_1,s_1),(i_2,s_1),(i_3,s_2),(i_4,s_3)\}$.  

In Example \ref{Ex3}, the matching obtained from the CT algorithm differs from the matching obtained from both the TTC algorithm and the FCT algorithm, $\mu^T=\{(i_1,s_2),(i_2,s_1),(i_3,s_1),(i_4,s_3)\}=\mu^F$. Notice that $\mu^D = \mu^C$.

Morrill (2015) shows that the Clinch and Trade mechanism (CT) is Pareto efficient and strategy-proof. Unlike the TTC mechanism, the CT mechanism is bossy, not group strategy-proof, and not independent of the order in which cycles are processed. 

Take for instance the matching $\mu_0 = \{(i_1,s_2),(i_2,s_1),(i_3,s_1),(i_4,s_3)\} = \mu^T$. We construct a farsighted improving from $\mu_0$  to $\mu^C = \{(i_1,s_1),(i_2,s_1),(i_3,s_2),(i_4,s_3)\} = \mu_2$ following the steps as in the proof of Theorem \ref{THCT}. 

First, we consider student $i_4$ who is the only student to be clinched in the first round of the clinching procedure of the CT. Since student $i_4$ is matched to school $s_3$ in both $\mu_0$ and $\mu_C$ she does not participate to any deviation and remains clinched to $s_3$ along the farsighted improving path. That is, $(i_4,s_3) \in \mu_l$, $l=0,1,2$. Next, we consider student $i_2$ who is the only student to be clinched in the second round of the clinching procedure. Since student $i_2$ is matched to school $s_1$ in both $\mu_0$ and $\mu_C$ she does not participate to any deviation and remains clinched to $s_2$ along the farsighted improving path. That is, $(i_2,s_1) \in \mu_l$, $l=0,1,2$. Next, we consider student $i_3$ who is the only student to be clinched in the third round of the clinching procedure. Looking forward towards $\mu^C$, the coalition $N_0=\{i_3,s_2\}$ deviates from $\mu_0$ so that student $i_3$ joins school $s_2$ to reach the matching $\mu_1 = \{(i_1,i_1),(i_2,s_1),(i_3,s_2),(i_4,s_3)\}$. Student $i_3$ is matched to her preferred school $s_2$ in $\mu_1$ and she has a higher priority than $i_1$ at $s_2$. By doing so, student $i_1$ is pushed out of school $s_2$. Next, there is one cycle between schools and students who are not clinched in the iterated clinching procedure: $i_1$ and $s_1$ form a cycle. So, the coalition $N_1=\{i_1,s_1\}$ deviates from $\mu_1$ so that student $i_1$ joins school $s_1$ to reach the matching $\mu_2 = \{(i_1,s_1),(i_2,s_1),(i_3,s_2),(i_4,s_3)\}=\mu^C$. Thus, $\mu^C \in \phi(\mu^T)$.

In fact, it holds in general that, from any $\mu \neq \mu^C$ there exists a farsighted improving path leading to $\mu^C$. Hence, $\{\mu^C\}$ is a farsighted stable set. 

\begin{atheorem}\label{THCT}
Let $\langle I,S,q,P,F\rangle$ be a school choice problem and $\mu^C$ be the matching obtained from the Clinch and Trade mechanism. The singleton set $\{\mu^C\}$ is a farsighted stable set.
\end{atheorem}

The proof of Theorem \ref{THCT} can be found in Appendix \ref{AP2}.

\subsection{Equitable Top Trading Cycles Algorithm} 

Hakimov and Kesten (2018) introduce the Equitable Top Trading Cycles mechanism for selecting a matching for each school problem by means of the Equitable Top Trading Cycles algorithm (ETTC). They show that the ETTC mechanism is Pareto-efficient and group strategy-proof and eliminates more avoidable justified envy situations than the TTC. Instead of allowing only the current highest priority students to participate in the trading process, the ETTC assigns all slots of each school $s$ to all the $q_s$ students with the highest priorities in each school, giving one slot to each student and endowing them with equal trading power. The terms of trade are next determined by a pointing rule specifying for each student-school pair which student-school pair should be pointed to among those who contain the remaining favorite school. In the ETTC, each student-school pair points to the pair containing the highest priority student for the school contained in the former pair in order to ensure that the students included in a cycle between two student-school pairs have the highest priority for their favorite schools among their competitors at that step of the trading market.

\begin{aexample}[Morrill, 2015]\label{Ex4} 
Consider a school choice problem $\langle I,S,q,P,F\rangle$ with $I=\{i_1,i_2,i_3,i_4\}$ and $S=\{s_1,s_2,s_3\}$. Students' preferences and schools' priorities and capacities are as follows.
\begin{center}
\begin{tabular}{cccc}
\multicolumn{4}{c}{Students}\\
\hline
$P_{i_1}$& $P_{i_2}$ & $P_{i_3}$ & $P_{i_4}$ \\
\hline
$s_1$ & $s_3$ & $s_2$ & $s_2$  \\
$s_2$ & $s_1$ & $s_1$  & $s_3$ \\
$s_3$ & $s_2$ & $s_3$  & $s_1$ 
\end{tabular}
~\qquad~
\begin{tabular}{lccc}
\multicolumn{4}{c}{Schools}\\
\hline
&     $F_{s_1}$& $F_{s_2}$ & $F_{s_3}$\\
\hline
$q_s$ & 2 & 1 & 1\\
\hline
&$i_2$ & $i_1$ & $i_1$\\
&$i_4$ & $i_2$ & $i_4$\\
&$i_1$ & $i_3$ & $i_2$\\
&$i_3$ & $i_4$ & $i_3$
    \end{tabular}
  \end{center}

\end{aexample}

Let $\mu^E$ be the matching obtained from the ETTC mechanism. A formal description of the ETTC algorithm can be found in Appendix \ref{AP3}. By means of Example \ref{Ex4} we illustrate the mechanism behind the ETTC algorithm. In the inheritance round of the first step of the ETTC algorithm, all seats are available to inherit and so, students are assigned to seats according to the priority orders $F$ to form the following student-school pairs: $(i_2,s_1)$, $(i_4,s_1)$, $(i_1,s_2)$ and $(i_1,s_3)$. Next, each student-school pair $(i,s)$ points to the student-pair $(i^{\prime},s^{\prime})$ such that $s^{\prime}$ is the top choice of student $i$ and $i^{\prime}$ has the highest priority for school $s$ among the students who are assigned a seat at school $s^{\prime}$ in the inheritance round. That is, $(i_2,s_1)$ points to $(i_1,s_3)$, $(i_4,s_1)$ points to $(i_1,s_2)$, $(i_1,s_2)$ points to $(i_2,s_1)$ and $(i_1,s_3)$ points to $(i_4,s_1)$. There is one cycle $(i_2,s_1) \mapsto (i_1,s_3) \mapsto (i_4,s_1) \mapsto (i_1,s_2) \mapsto (i_2,s_1)$.\footnote{There is at least one cycle. If some student appears in the same cycle or in different cycles with different schools, then she is definitely assigned a seat at her top choice among those schools while the other seats she was pointing to remain to be inherited in the next step.}  It leads to the following matches: $(i_1,s_1)$, $(i_2,s_3)$ and $(i_4,s_2)$. In the inheritance round of the second step of the ETTC algorithm, only one seat at school $s_1$ is available to inherit and so student $i_3$ is assigned to a seat at school $s_1$ to form the student-school pair $(i_3,s_1)$. Next, the student-school pair $(i_3,s_1)$ points to $(i_3,s_1)$ and the match $(i_3,s_1)$ is formed. We reach the matching obtained from the ETTC algorithm, $\mu^E = \{(i_1,s_1),(i_2,s_3),i_3,s_1),(i_4,s_2)\}$. 

In Example \ref{Ex4}, the matching obtained from the ETTC algorithm differs from the matching obtained from the TTC algorithm, $\mu^T=\{(i_1,s_1),(i_2,s_3),(i_3,s_2),(i_4,s_1)\}=\mu^F=\mu^C$. For completeness, $\mu^D=\{(i_1,s_1),(i_2,s_1),(i_3,s_2),(i_4,s_3)\}$ is the matching obtained from the DA algorithm.\footnote{In Example \ref{Ex1}, $\mu^T = \mu^F = \mu^C = \mu^E \neq \mu^D $. In Example \ref{Ex2}, $\mu^T = \mu^E \neq \mu^F = \mu^C = \mu^D$. In Example \ref{Ex3}, $\mu^T = \mu^E = \mu^F \neq \mu^C = \mu^D$.} 

Take for instance the matching $\mu_0 = \{(i_1,s_1),(i_2,s_3),(i_3,s_2),(i_4,s_1)\} = \mu^T$. We construct a farsighted improving from $\mu_0$  to $\mu^E = \{(i_1,s_1),(i_2,s_3),i_3,s_1),(i_4,s_2)\} = \mu_3$ following the steps as in the proof of Theorem \ref{THETTC}. In the first round of ETTC, there is only cycle between student-school pairs involving students $i_1,i_2,i_4$ and schools $s_1,s_2,s_3$. Looking forward towards $\mu^E$, the coalition $N_0=\{i_1,i_2,i_4,s_1,s_2\}$ deviates from $\mu_0$ so that student $i_1$ joins school $s_2$, student $i_2$ joins school $s_1$ and student $i_4$ joins school $s_1$ to reach the matching $\mu_1 = \{(i_1,s_2),(i_2,s_1),(i_3,i_3),(i_4,s_1)\}$. That is, each student is matched to the school from the pair student-school she belongs to. Next, the coalition $N_0=\{i_1,i_2,i_4\}$ deviates from $\mu_1$ so that student $i_1$ leaves school $s_2$, student $i_2$ leaves school $s_1$ and student $i_4$ leaves school $s_1$ to reach the matching $\mu_2 = \{(i_1,i_1),(i_2,i_2),(i_3,i_3),(i_4,i_4)\}$ where all students are unmatched. Next, the coalition $N_0=\{i_1,i_2,i_3,i_4,s_1,s_2,s_3\}$ deviates from $\mu_2$ so that student $i_1$ joins school $s_1$, student $i_2$ joins school $s_3$, student $i_3$ joins school $s_1$ and student $i_4$ joins school $s_2$ to reach the matching $\mu_3 = \{(i_1,s_1),(i_2,s_3),(i_3,s_1),(i_4,s_2)\} = \mu^E$. Thus, $\mu^E \in \phi(\mu^T)$. Again, it holds in general that, from any $\mu \neq \mu^E$ there exists a farsighted improving path leading to $\mu^E$. Thus, the matching obtained from the FCT algorithm also preserves the property of being stable once students are farsighted.

\begin{atheorem}\label{THETTC}
Let $\langle I,S,q,P,F\rangle$ be a school choice problem and $\mu^E$ be the matching obtained from the Equitable Top Trading Cycles mechanism. The singleton set $\{\mu^E\}$ is a farsighted stable set.
\end{atheorem}

The proof of Theorem \ref{THETTC} can be found in Appendix \ref{AP3}.

\section{Conclusion}\label{S7} 

We consider priority-based school choice problems. Once students are farsighted, the matching obtained from the TTC mechanism becomes stable: a singleton set consisting of the TTC matching is a farsighted stable set. However, the matching obtained from the DA mechanism may not belong to any farsighted stable set. Hence, the TTC mechanism provides an assignment that is not only Pareto efficient but also farsightedly stable. Moreover, looking forward three steps ahead is already sufficient for stabilizing the matching obtained from the TTC. Since the choice between the DA mechanism or the TTC mechanism usually depends on the priorities of the policy makers, farsightedness and Pareto efficiency may tip the balance in favor of TTC or one of its variations.

\section*{Acknowledgments}

Ana Mauleon and Vincent Vannetelbosch are, respectively, Research Director and Senior Research Associate of the National Fund for Scientific Research (FNRS). Financial support from the Fonds de la Recherche Scientifique - FNRS PDR research grant T.0143.18 is gratefully acknowledged. Ata Atay is a Serra H\'{u}nter Fellow (Professor Lector Serra H\'{u}nter). Ata Atay gratefully acknowledges financial support from the University of Barcelona through grant AS017672.

\appendix

\section{Appendix}

\subsection{First Clinch and Trade Algorithm}\label{AP1} 

Let $G_s$ be the set of students who are guaranteed admissions to school $s$. That is,
\begin{equation*}
G_s = \{i \in I \mid \# (\Phi(s,i)) < q_s \}.
\end{equation*}
Thus, a student $i$ is only guaranteed admission to a school $s$ if she initially has one of the $q_s$ highest rankings at school $s$. The First Clinch and Trade mechanism (Morrill, 2015) finds a matching by means of the following First Clinch and Trade algorithm (FCT).

\begin{itemize}
\item[Step $1$.] Set $q_s^{1} = q_s$ for all $s \in S$ where $q_s^{1}$ is the remaining capacity of school $s$ at Step $1$. 

First, each student $i \in I$ points to the school that is ranked first in $P_i$. If student $i$ is pointing to school $s$ and $i \in G_s$, then she is assigned to school $s$ and the capacity of school $s$ is reduced by one. Student $i$ is said to be clinched to school $s$. Let $m_1^0$ be all matches formed by students who clinch to some school:
\begin{equation*}
m_1^0= \{(i,s) \mid s \in S \text{, }  i \in G_s \text{ and } i\mapsto s \}.
\end{equation*}
The capacity of school $s$ is now equal to $q_s^{1} - \#\{(i,s) \mid i \in G_s \text{ and } i\mapsto s \}$ and $\{i \in I \mid s \in S \text{, } i \in G_s \text{ and } i\mapsto s \}$ is the set of students who  clinch to some school and are removed.

Second, each remaining student $i \in I \setminus \{i \in I \mid s \in S \text{, } i \in G_s \text{ and } i\mapsto s \}$ points again to the school $s$ that is ranked first in $P_i$. If there is no such school, then student $i$ points to herself and she forms a self-cycle. Each school $s \in S$ points to the student $j \in I$ that has the highest priority in $F_s$. If there exists a cycle, every student in a cycle is assigned a seat at the school she points to and she is removed. Each school (student) can be part of at most one cycle. Similarly, every student in a self-cycle is not assigned to any school and is removed. If a school $s$ is part of a cycle, then its remaining capacity $q_s^2$ is equal to $q_s^1 - \#\{(i,s) \mid i \in G_s \text{ and } i\mapsto s \} - 1$. If a school $s$ is not part of any cycle, then its remaining capacity $q_s^2$ remains equal to $q_s^1 - \#\{(i,s) \mid i \in G_s \text{ and } i\mapsto s \}$. If $q_s^2 = 0$, then school $s$ is removed. Let $C_1=\{c_1^1,c_1^2,...,c_1^{L_1}\}$ be the set of cycles in Step $1$ (where $L_1 \geq 1$ is the number of cycles in Step $1$). Let $m_1^l$ be all the matches from cycle $c_1^l$ that are formed in Step $1$ of the algorithm.
\begin{equation*}
m_1^l=\left\lbrace
\begin{matrix}
\{(i,s) \mid i,s \in c_1^l \text{ and } i\mapsto s \} & \text{if} & c_1^l \neq (j) \\
\{(j,j) \} & \text{if} & c_1^l = (j)
\end{matrix}\right.
\end{equation*}
Let $I_1$ be the set of students who are assigned to some school at Step $1$.  Let $M_1 = \cup_{l=0}^{L_1} m_1^l$ be all the matches between students and schools formed in Step $1$ of the algorithm.

\item[Step $k\geq 2$.] Notice that $q_s^k$ keeps track of how many seats are still available at the school at Step $k$ of the algorithm. 

First, each remaining student $i \in I \setminus \cup_{l=1}^{k-1}I_l$ points to the school $s$ that is ranked first in $P_i$ such that $q_s^k \geq 1$. If student $i$ is pointing to school $s$ and $i \in G_s$, then she is clinched and assigned to school $s$ and the capacity of school $s$ is reduced by one. Let $m_k^0$ be all such matches formed by students who clinch to some school in Step $k$:
\begin{equation*}
m_k^0= \{(i,s) \mid s \in S \text{, }  i \in G_s \cap (I \setminus \cup_{l=1}^{k-1}I_l) \text{ and } i\mapsto s \}.
\end{equation*}
The capacity of school $s$ is now equal to $q_s^{k} - \#\{(i,s) \mid i \in G_s \cap (I \setminus \cup_{l=1}^{k-1}I_l) \text{ and } i\mapsto s \}$ and $\{i \in I \mid s \in S \text{, } i \in G_s \cap (I \setminus \cup_{l=1}^{k-1}I_l) \text{ and } i\mapsto s \}$ is the set of students who  clinch to some school in Step $k$. 

Second, each remaining student $i \in (I \setminus \cup_{l=1}^{k-1}I_l) \setminus \{i \in I \mid s \in S \text{, } i \in G_s \cap (I \setminus \cup_{l=1}^{k-1}I_l) \text{ and } i\mapsto s \}$ points again to the school $s$ that is ranked first in $P_i$ such that $q_s^{k} \geq 1$. If there is no such school, then student $i$ points to herself and she forms a self-cycle. Each school $s \in S$ such that $q_s^{k} \geq 1$ points to the student $j \in (I \setminus \cup_{l=1}^{k-1}I_l) $ that has the highest priority in $F_s$. If there exists a cycle, every student in a cycle is assigned a seat at the school she points to and she is removed. Similarly, every student in a self-cycle is not assigned to any school and is removed. If a school $s$ is part of a cycle, then its remaining capacity $q_s^{k+1}$ is equal to $q_s^{k} - \#\{(i,s) \mid i \in G_s \cap (I \setminus \cup_{l=1}^{k-1}I_l) \text{ and } i\mapsto s \} - 1$. If a school $s$ is not part of any cycle, then its remaining capacity $q_s^{k+1}$ remains equal to $q_s^{k} - \#\{(i,s) \mid i \in G_s \cap (I \setminus \cup_{l=1}^{k-1}I_l) \text{ and } i\mapsto s \}$. If $q_s^{k+1} = 0$, then school $s$ is removed. Let $C_k=\{c_k^1,c_k^2,...,c_k^{L_k}\}$ be the set of cycles in Step $k$ (where $L_k \geq 1$ is the number of cycles in Step $k$). Let $m_k^l$ be all the matches from cycle $c_k^l$ that are formed in Step $k$ of the algorithm.
\begin{equation*}
m_k^l=\left\lbrace
\begin{matrix}
\{(i,s) \mid i,s \in c_k^l \text{ and } i\mapsto s \} & \text{if} & c_k^l \neq (j) \\
\{(j,j) \} & \text{if} & c_k^l = (j)
\end{matrix}\right.
\end{equation*}
Let $M_k = \cup_{l=0}^{L_k} m_k^l$ be all the matches between students and schools formed in Step $k$ of the algorithm. Let $I_k$ be the set of students who are assigned to some school at Step $k$. 

\item[End.] The algorithm stops when all students have been removed. Let $\bar{k}$ be the step at which the algorithm stops. Let $\mu^F$ denote the matching obtained from the First Clinch and Trade mechanism and it is given by $\mu^F = \cup_{k=1}^{\bar{k}} M_k$. 
\end{itemize}

\noindent \textbf{Proof of Theorem \ref{THFCT}}

Since $\{\mu^F\}$ is a singleton set, internal stability (IS) is satisfied. (ES) Take any matching $\mu \neq \mu^F$, we need to show that $\phi(\mu) \ni \mu^F$. We build in steps a farsighted improving path from $\mu$ to $\mu^F$.
\begin{itemize}
\item[Step 1.0.]If $m_1^0 \subseteq \mu$ and $C_1 \neq \emptyset$ then go to Step 1.1 with $\mu_{1,0}^{\prime} = \mu$. If $m_1^0 \subseteq \mu$ and $C_1 = \emptyset$ then go to Step 1.End with $\mu_{1,L_1}^{\prime \prime \prime} = \mu$. Notice that it is not excluded that $m_1^0 = \emptyset$. Let $\Lambda_1= \#\{(i,s) \notin \mu \mid (i,s) \in m_1^0 \}$ be the number of students who are not yet matched to their preferred school in $\mu$ and are guaranteed admissions to their preferred school. If $m_1^0 \nsubseteq \mu$ then $\mu_{1,0}^{\prime} = \mu + \{(i,s) \notin \mu \mid (i,s) \in m_1^0 \} - \{(j,s) \in \mu \mid \Lambda_j^s (\mu) < \Lambda_1 - q_s + \#\mu(s)\}$ where $\Lambda_j^s (\mu) = \#\{l \in I \mid (l,s) \in \mu \text{ and } F_s(l) > F_s(j)\}$ is the number of students who are matched to school $s$ in $\mu$ and have a lower priority than student $j$. We reach $\mu_{1,0}^{\prime }$ with $m_1^0 \subseteq \mu_{1,0}^{\prime }$. If $C_1 \neq \emptyset$, then go to Step 1.1. Otherwise, go to Step 1.End with $\mu_{1,L_1}^{\prime \prime \prime} = \mu_{1,0}^{\prime }$.

\item[Step 1.1.]If $m_1^1 \subseteq \mu_{1,0}^{\prime}$ and $1 \neq L_1$ then go to Step 1.2 with $\mu_{1,1}^{\prime \prime \prime} = \mu_{1,0}^{\prime}$. If $m_1^1 \subseteq \mu_{1,0}^{\prime}$ and $1 = L_1$ then go to Step 1.End with $\mu_{1,L_1}^{\prime \prime \prime} = \mu_{1,0}^{\prime}$. If $m_1^1 \nsubseteq \mu_{1,0}^{\prime}$ then $\mu_{1,1}^{\prime} = \mu_{1,0}^{\prime} - \{(i,\mu_{1,0}^{\prime}(i)) \mid (i,\mu^F(i)) \in m_1^1 \text{ and } \mu_{1,0}^{\prime}(i) \neq i \} + \{(i,s) \mid i,s \in c_1^1 \text{ and } s \mapsto i\} - \{(j,s) \in \mu_{1,0}^{\prime} \mid s \in c_1^1 \text{, } \mu_{1,0}^{\prime}(s) \cap c_1^1 = \emptyset \text{, } \# \mu_{1,0}^{\prime}(s)=q_s \text{ and } F_s(j)>F_s(l) \text{ for all } l \in \mu_{1,0}^{\prime}(s), l \neq j \}$. Next $\mu_{1,1}^{\prime \prime} = \mu_{1,1}^{\prime} - \{(i,s) \mid i,s \in c_1^1 \text{ and } s \mapsto i\}$. Next $\mu_{1,1}^{\prime \prime \prime} = \mu_{1,1}^{\prime \prime} + \{(i,s) \mid i,s \in c_1^1 \text{ and } i \mapsto s\}$. We reach $\mu_{1,1}^{\prime \prime \prime}$ with $m_1^1 \subseteq \mu_{1,1}^{\prime \prime \prime}$. If $1 \neq L_1$, then go to Step 1.2. Otherwise, go to Step 1.End with $\mu_{1,L_1}^{\prime \prime \prime} = \mu_{1,1}^{\prime \prime \prime}$.

\item[Step 1.$k$.]($k>1$) If $m_1^k \subseteq \mu_{1,k-1}^{\prime \prime \prime}$ and $k \neq L_1$ then go to Step 1.$k$+1 with $\mu_{1,k}^{\prime \prime \prime} = \mu_{1,k-1}^{\prime \prime \prime}$. If $m_1^k \subseteq \mu_{1,k-1}^{\prime \prime \prime}$ and $k = L_1$ then go to Step 1.End with $\mu_{1,L_1}^{\prime \prime \prime} = \mu_{1,k-1}^{\prime \prime \prime}$. If $m_1^k \nsubseteq \mu_{1,k-1}^{\prime \prime \prime}$ then $\mu_{1,k}^{\prime} = \mu_{1,k-1}^{\prime \prime \prime} - \{(i,\mu_{1,k-1}^{\prime \prime \prime}(i)) \mid (i,\mu^F(i)) \in m_1^k \text{ and } \mu_{1,k-1}^{\prime \prime \prime}(i) \neq i\} + \{(i,s) \mid i,s \in c_1^k \text{ and } s \mapsto i\} - \{(j,s) \in \mu_{1,k-1}^{\prime \prime \prime} \mid s \in c_1^k \text{, } \mu_{1,k-1}^{\prime \prime \prime}(s) \cap c_1^k = \emptyset \text{, } \# \mu_{1,k-1}^{\prime \prime \prime}(s)=q_s \text{ and } F_s(j)>F_s(l) \text{ for all } l \in \mu_{1,k-1}^{\prime \prime \prime}(s), l \neq j \}$. Next $\mu_{1,k}^{\prime \prime} = \mu_{1,k}^{\prime} - \{(i,s) \mid i,s \in c_1^k \text{ and } s \mapsto i\}$. Next $\mu_{1,k}^{\prime \prime \prime} = \mu_{1,k}^{\prime \prime} + \{(i,s) \mid i,s \in c_1^k \text{ and } i \mapsto s\}$. We reach $\mu_{1,k}^{\prime \prime \prime}$ with $m_1^k \subseteq \mu_{1,k}^{\prime \prime \prime}$. If $k \neq L_1$, then go to Step 1.$k$+1. Otherwise, go to Step 1.End with $\mu_{1,L_1}^{\prime \prime \prime} = \mu_{1,k}^{\prime \prime \prime}$.

\item[Step 1.End.] We have reached $\mu_{1,L_1}^{\prime \prime \prime}$ with $\cup_{l=0}^{L_1} m_1^{l} = M_1 \subseteq \mu_{1,L_1}^{\prime \prime \prime}$. If $\mu_{1,L_1}^{\prime \prime \prime} = \mu^F$ then the process ends. Otherwise, go to Step 2.0.

\item[Step 2.0.]If $m_2^0 \subseteq \mu_{1,L_1}^{\prime \prime \prime}$ and $C_2 \neq \emptyset$ then go to Step 2.1 with $\mu_{2,0}^{\prime} = \mu_{1,L_1}^{\prime \prime \prime}$. If $m_2^0 \subseteq \mu_{1,L_1}^{\prime \prime \prime}$ and $C_2 = \emptyset$ then go to Step 2.End with $\mu_{2,L_2}^{\prime \prime \prime} = \mu_{1,L_1}^{\prime \prime \prime}$. Given $M_1 \subseteq \mu_{1,L_1}^{\prime \prime \prime}$, let $\Lambda_2= \#\{(i,s) \notin \mu_{1,L_1}^{\prime \prime \prime} \mid (i,s) \in m_2^0 \}$ be the number of students who are not yet matched to their preferred school in $\mu_{1,L_1}^{\prime \prime \prime}$ and are guaranteed admissions to their preferred school. If $m_2^0 \nsubseteq \mu_{1,L_1}^{\prime \prime \prime}$ then $\mu_{2,0}^{\prime} = \mu_{1,L_1}^{\prime \prime \prime} + \{(i,s) \notin \mu_{1,L_1}^{\prime \prime \prime} \mid (i,s) \in m_2^0 \} - \{(j,s) \in \mu_{1,L_1}^{\prime \prime \prime} \mid \Lambda_j^s (\mu_{1,L_1}^{\prime \prime \prime}) < \Lambda_2 - q_s + \#\mu_{1,L_1}^{\prime \prime \prime}(s)\}$ where $\Lambda_j^s (\mu_{1,L_1}^{\prime \prime \prime}) = \#\{l \in I \mid (l,s) \in \mu_{1,L_1}^{\prime \prime \prime} \text{ and } F_s(l) > F_s(j)\}$ is the number of students who are matched to school $s$ in $\mu_{1,L_1}^{\prime \prime \prime}$ and have a lower priority than student $j$. We reach $\mu_{2,0}^{\prime }$ with $m_2^0 \subseteq \mu_{2,0}^{\prime }$. If $C_2 \neq \emptyset$, then go to Step 2.1. Otherwise, go to Step 2.End with $\mu_{2,L_2}^{\prime \prime \prime} = \mu_{2,0}^{\prime }$.

\item[Step 2.1.]If $m_2^1 \subseteq \mu_{2,0}^{\prime}$ and $1 \neq L_2$ then go to Step 2.2 with $\mu_{2,1}^{\prime \prime \prime} = \mu_{2,0}^{\prime}$. If $m_2^1 \subseteq \mu_{2,0}^{\prime}$ and $1 = L_2$ then go to Step 2.End with $\mu_{2,L_2}^{\prime \prime \prime} = \mu_{2,0}^{\prime}$. If $m_2^1 \nsubseteq \mu_{2,0}^{\prime}$ then $\mu_{2,1}^{\prime} = \mu_{2,0}^{\prime} - \{(i,\mu_{2,0}^{\prime}(i)) \mid (i,\mu^F(i)) \in m_2^1 \text{ and } \mu_{2,0}^{\prime}(i) \neq i \} + \{(i,s) \mid i,s \in c_2^1 \text{ and } s \mapsto i\} - \{(j,s) \in \mu_{2,0}^{\prime} \mid s \in c_2^1 \text{, } \mu_{2,0}^{\prime}(s) \cap c_2^1 = \emptyset \text{, } \# \mu_{2,0}^{\prime}(s)=q_s \text{ and } F_s(j)>F_s(l) \text{ for all } l \in \mu_{2,0}^{\prime}(s), l \neq j \}$. Next $\mu_{2,1}^{\prime \prime} = \mu_{2,1}^{\prime} - \{(i,s) \mid i,s \in c_2^1 \text{ and } s \mapsto i\}$. Next $\mu_{2,1}^{\prime \prime \prime} = \mu_{2,1}^{\prime \prime} + \{(i,s) \mid i,s \in c_2^1 \text{ and } i \mapsto s\}$. We reach $\mu_{2,1}^{\prime \prime \prime}$ with $m_2^1 \subseteq \mu_{2,1}^{\prime \prime \prime}$. If $1 \neq L_2$, then go to Step 2.2. Otherwise, go to Step 2.End with $\mu_{2,L_2}^{\prime \prime \prime} = \mu_{2,1}^{\prime \prime \prime}$.

\item[Step 2.$k$.]($k>1$) If $m_2^k \subseteq \mu_{2,k-1}^{\prime \prime \prime}$ and $k \neq L_2$ then go to Step 2.$k$+1 with $\mu_{2,k}^{\prime \prime \prime} = \mu_{2,k-1}^{\prime \prime \prime}$. If $m_2^k \subseteq \mu_{2,k-1}^{\prime \prime \prime}$ and $k = L_2$ then go to Step 2.End with $\mu_{2,L_2}^{\prime \prime \prime} = \mu_{2,k-1}^{\prime \prime \prime}$. If $m_2^k \nsubseteq \mu_{2,k-1}^{\prime \prime \prime}$ then $\mu_{2,k}^{\prime} = \mu_{2,k-1}^{\prime \prime \prime} - \{(i,\mu_{2,k-1}^{\prime \prime \prime}(i)) \mid (i,\mu^F(i)) \in m_2^k \text{ and } \mu_{2,k-1}^{\prime \prime \prime}(i) \neq i \} + \{(i,s) \mid i,s \in c_2^k \text{ and } s \mapsto i\} - \{(j,s) \in \mu_{2,k-1}^{\prime \prime \prime} \mid s \in c_2^k \text{, } \mu_{2,k-1}^{\prime \prime \prime}(s) \cap c_2^k = \emptyset \text{, } \# \mu_{2,k-1}^{\prime \prime \prime}(s)=q_s \text{ and } F_s(j)>F_s(l) \text{ for all } l \in \mu_{2,k-1}^{\prime \prime \prime}(s), l \neq j \}$. Next $\mu_{2,k}^{\prime \prime} = \mu_{2,k}^{\prime} - \{(i,s) \mid i,s \in c_2^k \text{ and } s \mapsto i\}$. Next $\mu_{2,k}^{\prime \prime \prime} = \mu_{2,k}^{\prime \prime} + \{(i,s) \mid i,s \in c_2^k \text{ and } i \mapsto s\}$. We reach $\mu_{2,k}^{\prime \prime \prime}$ with $m_2^k \subseteq \mu_{2,k}^{\prime \prime \prime}$. If $k \neq L_2$, then go to Step 2.$k$+1. Otherwise, go to Step 2.End with $\mu_{2,L_2}^{\prime \prime \prime} = \mu_{2,k}^{\prime \prime \prime}$.

\item[Step 2.End.] We have reached $\mu_{2,L_2}^{\prime \prime \prime}$ with $M_1 \cup M_2 \subseteq \mu_{2,L_2}^{\prime \prime \prime}$ where $\cup_{l=0}^{L_1} m_1^{l} = M_1$ and $\cup_{l=0}^{L_2} m_2^{l} = M_2$. If $\mu_{2,L_2}^{\prime \prime \prime} = \mu^F$ then the process ends. Otherwise, go to Step 3.0.

\item[End.]The process goes on until we reach $\mu_{\bar{k},L_{\bar{k}}}^{\prime \prime \prime} = \cup_{k=1}^{\bar{k}} M_k = \mu^F$.

\end{itemize}
\begin{flushright}
\qedsymbol
\end{flushright}

\subsection{Clinch and Trade Algorithm}\label{AP2} 

The Clinch and Trade mechanism (Morrill, 2015) finds a matching by means of the following Clinch and Trade algorithm (CT).

\begin{itemize}
\item[Step $1$.] Set $q_s^{1} = q_s$ for all $s \in S$ where $q_s^{1}$ is the remaining capacity of school $s$ at Step $1$. 

\item[$1$.A.] Let $q_s^{1,1} = q_s^{1}$ for all $s \in S$.

In the first round of the clinching procedure, each student $i \in I$ points to the school that is ranked first in $P_i$. If student $i$ is pointing to school $s$ and $i \in G_s^{1,1} = \{i \in I \mid \# (\Phi(s,i)) < q_s^{1,1} \}$, then she is assigned to school $s$ and the capacity of school $s$ is reduced by one. Student $i$ is said to be clinched to school $s$. Let $I_0^{1,1} = \{i \in I \mid i \in G_s^{1,1} \text{ and } i\mapsto s \}$ be the set of students who clinch to some school in the first round and let $m_1^{0,1}$ be all the matches formed by students belonging to $I_0^{1,1}$. Let $q_s^{1,2} = q_s^{1,1} - \#\{(i,s) \mid i \in G_s^{1,1} \text{ and } i\mapsto s \}$ be the capacity of school $s$ at the end of the first round. Whenever a student is removed, the rankings of all schools are adjusted accordingly. In the second round of the clinching procedure, each remaining student $i \in I \setminus I_0^{1,1}$ points to the school that is ranked first in $P_i$. If student $i$ is pointing to school $s$ and $i \in G_s^{1,2} = \{i \in I \setminus I_0^{1,1} \mid \# (\Phi(s,i) \setminus I_0^{1,1}) < q_s^{1,2} \}$, then she is assigned to school $s$ and the capacity of school $s$ is reduced by one. Let $I_0^{1,2} = \{i \in I \setminus I_0^{1,1} \mid i \in G_s^{1,2} \text{ and } i\mapsto s \}$ be the set of students who clinch to some school in the second round and let $m_1^{0,2}$ be all the matches formed by students belonging to $I_0^{1,2}$. Let $q_s^{1,3} = q_s^{1,2} - \#\{(i,s) \mid i \in G_s^{1,2} \text{ and } i\mapsto s \}$ be the capacity of school $s$ at the end of the second round. Whenever a student is removed, the rankings of all schools are adjusted accordingly. Let $I_0^{1,0} = \emptyset$. For $r \geq 1$, let
\begin{equation*}
\begin{split}
G_s^{1,r} & = \{i \in I \setminus \cup_{l=0}^{r-1} I_0^{1,l}  \mid \# ( \Phi(s,i) \setminus \cup_{l=0}^{r-1} I_0^{1,l} ) < q_s^{1,r} \}, \\
q_s^{1,r+1} & = q_s^{1,r} - \#\{(i,s) \mid i \in G_s^{1,r} \text{ and } i\mapsto s \}, \\
I_0^{1,r} & = \{i \in I \setminus \cup_{l=0}^{r-1} I_0^{1,l}  \mid i \in G_s^{1,r} \text{ and } i\mapsto s \}, \\
m_1^{0,r} & = \{(i,s) \mid i \in G_s^{1,r} \text{ and } i\mapsto s \}.
\end{split}
\end{equation*}

The clinching procedure is iterated until $I_0^{1,r_1^{\prime}} \neq \emptyset$ while $I_0^{1,r_1^{\prime}+1} = \emptyset$ for some $r_1^{\prime}$. Then, $M_1^0 = \cup_{r=1}^{r_1^{\prime}} m_1^{0,r}$ are all the matches obtained from iterating the clinching procedure in Step 1.A. 

\item[$1$.B] Each remaining student $i \in I \setminus \cup_{l=0}^{r_1^{\prime}} I_0^{1,l} $ points to the school $s$ that is ranked first in $P_i$ such that $q_s^{1,r_1^{\prime}+1} \geq 1$. If there is no such school, then student $i$ points to herself and she forms a self-cycle. Each school $s \in S$ such that $q_s^{1,r_1^{\prime}+1} \geq 1$ points to the student $j \in I \setminus \cup_{l=0}^{r_1^{\prime}} I_0^{1,l}$ that has the highest priority in $F_s$. If there exists a cycle, every student in a cycle is assigned a seat at the school she points to and she is removed. Each school (student) can be part of at most one cycle. Similarly, every student in a self-cycle is not assigned to any school and is removed. If a school $s$ is part of a cycle, then its remaining capacity $q_s^2$ is equal to $q_s^{1,r_1^{\prime}+1} - 1$. If a school $s$ is not part of any cycle, then its remaining capacity $q_s^2$ remains equal to $q_s^{1,r_1^{\prime}+1}$. If $q_s^2 = 0$, then school $s$ is removed. Let $C_1=\{c_1^1,c_1^2,...,c_1^{L_1}\}$ be the set of cycles in Step $1$.B (where $L_1 \geq 1$ is the number of cycles in Step $1$.B). Let $m_1^l$ be all the matches from cycle $c_1^l$ that are formed in Step $1$.B of the algorithm.
\begin{equation*}
m_1^l=\left\lbrace
\begin{matrix}
\{(i,s) \mid i,s \in c_1^l \text{ and } i\mapsto s \} & \text{if} & c_1^l \neq (j) \\
\{(j,j) \} & \text{if} & c_1^l = (j)
\end{matrix}\right.
\end{equation*}
Let $I_1^1$ be the set of students who are assigned to some school and let $M_1^1 = \cup_{l=1}^{L_1} m_1^l$ be all the matches between students and schools formed in Step $1$.B of the algorithm. Let $I_1 = I_1^1 \cup (\cup_{l=0}^{r_1^{\prime}} I_0^{1,l})$ be the set of students who are assigned to some school in Step $1$ of the algorithm and let $M_1 = M_1^0 \cup M_1^1 $ be all the matches formed between students and schools.

\item[Step $k\geq 2$] Notice that $q_s^k$ keeps track of how many seats are still available at the school at step $k$ of the algorithm. Let $\widetilde{I}_{k-1} = I \setminus (\cup_{l=1}^{k-1} I_l)$ be the set of students who are not yet assigned at the end of Step $k-1$ of the algorithm.

\item[$k$.A.] Let $q_s^{k,1} = q_s^{k}$ for all $s \in S$. Let $\widehat{I}_k = \{i \in \widetilde{I}_{k-1} \mid i\mapsto s \text{ in Step } k-1.\text{B} \text{ and } q_s^{k,1} \geq 1 \}$ be the set of students who were pointing to some school in Step $k-1$.B that still has available capacity in Step $k$.A (i.e. $q_s^{k,1} \geq 1$). Each student $i \in \widehat{I}_k$ does not participate to the clinching procedure. In the first round of the clinching procedure, each student $i \in \widetilde{I}_{k-1} \setminus \widehat{I}_k$ points to the school that is ranked first in $P_i$. If student $i$ is pointing to school $s$ and $i \in G_s^{k,1} = \{i \in \widetilde{I}_{k-1} \setminus \widehat{I}_k \mid \# (\Phi(s,i) \cap \widetilde{I}_{k-1}) < q_s^{k,1} \}$, then she is assigned to school $s$ and the capacity of school $s$ is reduced by one. Student $i$ is said to be clinched to school $s$. Let $I_0^{k,1} = \{i \in \widetilde{I}_{k-1} \setminus \widehat{I}_k \mid i \in G_s^{k,1} \text{ and } i\mapsto s \}$ be the set of students who clinch to some school in the first round and let $m_k^{0,1}$ be all the matches formed by students belonging to $I_0^{k,1}$. Let $q_s^{k,2} = q_s^{k,1} - \#\{(i,s) \mid i \in G_s^{k,1} \text{ and } i\mapsto s \}$ be the capacity of school $s$ at the end of the first round. Whenever a student is removed, the rankings of all schools are adjusted accordingly. In the second round of the clinching procedure, each remaining student $i \in \widetilde{I}_{k-1} \setminus (\widehat{I}_k \cup I_0^{k,1}) $ points to the school that is ranked first in $P_i$. If student $i$ is pointing to school $s$ and $i \in G_s^{k,2} = \{i \in \widetilde{I}_{k-1} \setminus (\widehat{I}_k \cup I_0^{k,1}) \mid \# (\Phi(s,i) \cap (\widetilde{I}_{k-1} \setminus I_0^{k,1}) < q_s^{k,2} \}$, then she is assigned to school $s$ and the capacity of school $s$ is reduced by one. Let $I_0^{k,2} = \{i \in \widetilde{I}_{k-1} \setminus (\widehat{I}_k \cup I_0^{k,1}) \mid i \in G_s^{k,2} \text{ and } i\mapsto s \}$ be the set of students who clinch to some school in the second round and let $m_k^{0,2}$ be all the matches formed by students belonging to $I_0^{k,2}$. Let $q_s^{k,3} = q_s^{k,2} - \#\{(i,s) \mid i \in G_s^{k,2} \text{ and } i\mapsto s \}$ be the capacity of school $s$ at the end of the second round. Whenever a student is removed, the rankings of all schools are adjusted accordingly. Let $I_0^{k,0} = \emptyset$. For $r \geq 1$, let
\begin{equation*}
\begin{split}
G_s^{k,r} & = \{i \in \widetilde{I}_{k-1} \setminus (\widehat{I}_k \cup (\cup_{l=0}^{r-1} I_0^{k,l}))  \mid \# ( \Phi(s,i) \setminus \cup_{l=0}^{r-1} I_0^{k,l} ) < q_s^{k,r} \}, \\
q_s^{k,r+1} & = q_s^{k,r} - \#\{(i,s) \mid i \in G_s^{k,r} \text{ and } i\mapsto s \}, \\
I_0^{k,r} & = \{i \in \widetilde{I}_{k-1} \setminus (\widehat{I}_k \cup (\cup_{l=0}^{r-1} I_0^{k,l}))  \mid i \in G_s^{k,r} \text{ and } i\mapsto s \}, \\
m_k^{0,r} & = \{(i,s) \mid i \in G_s^{k,r} \text{ and } i\mapsto s \}.
\end{split}
\end{equation*}

The clinching procedure is iterated until $I_0^{k,r_k^{\prime}} \neq \emptyset$ while $I_0^{k,r_k^{\prime}+1} = \emptyset$ for some $r_k^{\prime}$. Then, $M_k^0 = \cup_{r=1}^{r_k^{\prime}} m_k^{0,r}$ are all the matches obtained from iterating the clinching procedure in Step $k$.A.

\item[$k$.B] Each remaining student $i \in \widetilde{I}_{k-1} \setminus \cup_{l=0}^{r_k^{\prime}} I_0^{k,l} $ points to the school $s$ that is ranked first in $P_i$ such that $q_s^{k,r_k^{\prime}+1} \geq 1$. If there is no such school, then student $i$ points to herself and she forms a self-cycle. Each school $s \in S$ such that $q_s^{k,r_k^{\prime}+1} \geq 1$ points to the student $j \in \widetilde{I}_{k-1} \setminus \cup_{l=0}^{r_k^{\prime}} I_0^{k,l}$ that has the highest priority in $F_s$. If there exists a cycle, every student in a cycle is assigned a seat at the school she points to and she is removed. Each school (student) can be part of at most one cycle. Similarly, every student in a self-cycle is not assigned to any school and is removed. If a school $s$ is part of a cycle, then its remaining capacity $q_s^{k+1}$ is equal to $q_s^{k,r_k^{\prime}+1} - 1$. If a school $s$ is not part of any cycle, then its remaining capacity $q_s^{k+1}$ remains equal to $q_s^{k,r_k^{\prime}+1}$. If $q_s^{k+1} = 0$, then school $s$ is removed. Let $C_k=\{c_k^1,c_k^2,...,c_k^{L_k}\}$ be the set of cycles in Step $k$.B (where $L_k \geq 1$ is the number of cycles in Step $k$.B). Let $m_k^l$ be all the matches from cycle $c_k^l$ that are formed in Step $k$.B of the algorithm.
\begin{equation*}
m_k^l=\left\lbrace
\begin{matrix}
\{(i,s) \mid i,s \in c_k^l \text{ and } i\mapsto s \} & \text{if} & c_k^l \neq (j) \\
\{(j,j) \} & \text{if} & c_k^l = (j)
\end{matrix}\right.
\end{equation*}
Let $I_k^1$ be the set of students who are assigned to some school and let $M_k^1 = \cup_{l=1}^{L_k} m_k^l$ be all the matches between students and schools formed in Step $k$.B of the algorithm. Let $I_k = I_k^1 \cup (\cup_{l=0}^{r_k^{\prime}} I_0^{k,l})$ be the set of students who are assigned to some school in Step $k$ of the algorithm and let $M_k = M_k^0 \cup M_k^1 $ be all the matches formed between students and schools.

\item[End] The algorithm stops when all students have been removed. Let $\bar{k}$ be the step at which the algorithm stops. Let $\mu^C$ denote the matching obtained from the Clinch and Trade mechanism and it is given by $\mu^C = \cup_{k=1}^{\bar{k}} M_k$. 
\end{itemize}

\noindent \textbf{Proof of Theorem \ref{THCT}}

Since $\{\mu^C\}$ is a singleton set, internal stability (IS) is satisfied. (ES) Take any matching $\mu \neq \mu^C$, we need to show that $\phi(\mu) \ni \mu^C$. We build in steps a farsighted improving path from $\mu$ to $\mu^C$.
\begin{itemize}
\item[Step 1.A.1.]If $m_1^{0,1} \subseteq \mu$ and $r_1^{\prime} \neq 1$ then go to Step 1.A.2 with $\mu_{1,1} = \mu$. If $m_1^{0,1} \subseteq \mu$, $r_1^{\prime}=1$ and $C_1 \neq \emptyset$ then go to Step 1.B.1 with $\mu_{1,r_1^{\prime}} = \mu$. If $m_1^{0,1} \subseteq \mu$, $r_1^{\prime}=1$ and $C_1 = \emptyset$ then go to Step 1.End with $\mu_{1,L_1}^{\prime \prime \prime} = \mu$. It is not excluded that $m_1^{0,1} = \emptyset$. Let $\Lambda_{1,1}(s)= \#\{(i,s^{\prime}) \notin \mu \mid (i,s^{\prime}) \in m_1^{0,1} \text{ and } s^{\prime} = s\}$ be the number of students who are not yet matched to their preferred school $s$ in $\mu$ and are guaranteed admissions to their preferred school $s$. If $m_1^{0,1} \nsubseteq \mu$ then $\mu_{1,1} = \mu + \{(i,s) \notin \mu \mid (i,s) \in m_1^{0,1} \} - \{(j,s) \in \mu \mid \Lambda_j^s (\mu) < \Lambda_{1,1}(s) - q_s + \#\mu(s)\}$ where $\Lambda_j^s (\mu) = \#\{l \in I \mid (l,s) \in \mu \text{ and } F_s(l) > F_s(j)\}$ is the number of students who are matched to school $s$ in $\mu$ and have a lower priority than student $j$. We reach $\mu_{1,1}$ with $m_1^{0,1} \subseteq \mu_{1,1}$. If $r_1^{\prime} \neq 1$ then go to Step 1.A.2. If $r_1^{\prime} = 1$ and $C_1 \neq \emptyset$, then go to Step 1.B.1. Otherwise, go to Step 1.End with $\mu_{1,L_1}^{\prime \prime \prime} = \mu_{1,1}$.

\item[Step 1.A.$k$.]($k>1$) If $m_1^{0,k} \subseteq \mu_{1,k-1}$ and $r_1^{\prime} \neq k$ then go to Step 1.A.$k+1$ with $\mu_{1,k} = \mu_{1,k-1}$. If $m_1^{0,k} \subseteq \mu_{1,k-1}$, $r_1^{\prime}=k$ and $C_1 \neq \emptyset$ then go to Step 1.B.1 with $\mu_{1,r_1^{\prime}} = \mu_{1,k-1}$. If $m_1^{0,k} \subseteq \mu_{1,k-1}$, $r_1^{\prime}=k$ and $C_1 = \emptyset$ then go to Step 1.End with $\mu_{1,L_1}^{\prime \prime \prime} = \mu_{1,k-1}$. It is not excluded that $m_1^{0,k} = \emptyset$. Given the matches $\cup_{r=1}^{k-1} m_1^{0,r}  \subseteq \mu_{1,k-1}$ remain fixed, let $\Lambda_{1,k}(s)= \#\{(i,s^{\prime}) \notin \mu_{1,k-1} \mid (i,s^{\prime}) \in m_1^{0,k} \text{ and } s^{\prime} = s\}$ be the number of students who are not yet matched to their preferred school $s$ in $\mu_{1,k-1}$ and are guaranteed admissions to their preferred school $s$. If $m_1^{0,k} \nsubseteq \mu_{1,k-1}$ then $\mu_{1,k} = \mu_{1,k-1} + \{(i,s) \notin \mu_{1,k-1} \mid (i,s) \in m_1^{0,k} \} - \{(j,s) \in \mu_{1,k-1} \mid \Lambda_j^s (\mu_{1,k-1}) < \Lambda_{1,k}(s) - q_s + \#\mu_{1,k-1}(s)\}$ where $\Lambda_j^s (\mu_{1,k-1}) = \#\{l \in I \mid (l,s) \in \mu_{1,k-1} \text{ and } F_s(l) > F_s(j)\}$ is the number of students who are matched to school $s$ in $\mu_{1,k-1}$ and have a lower priority than student $j$. We reach $\mu_{1,k}$ with $m_1^{0,k} \subseteq \mu_{1,k}$. If $r_1^{\prime} \neq k$ then go to Step 1.A.$k+1$. If $r_1^{\prime} = k$ and $C_1 \neq \emptyset$, then go to Step 1.B.1. Otherwise, go to Step 1.End with $\mu_{1,L_1}^{\prime \prime \prime} = \mu_{1,k}$.

\item[Step 1.B.1.]If $m_1^1 \subseteq \mu_{1,r_1^{\prime}}$ and $1 \neq L_1$ then go to Step 1.B.2 with $\mu_{1,1}^{\prime \prime \prime} = \mu_{1,r_1^{\prime}}$. If $m_1^1 \subseteq \mu_{1,r_1^{\prime}}$ and $1 = L_1$ then go to Step 1.End with $\mu_{1,L_1}^{\prime \prime \prime} = \mu_{1,r_1^{\prime}}$. If $m_1^1 \nsubseteq \mu_{1,r_1^{\prime}}$ then $\mu_{1,1}^{\prime} = \mu_{1,r_1^{\prime}} - \{(i,\mu_{1,r_1^{\prime}}(i)) \mid (i,\mu^C(i)) \in m_1^1 \text{ and } \mu_{1,r_1^{\prime}}(i) \neq i \} + \{(i,s) \mid i,s \in c_1^1 \text{ and } s \mapsto i\} - \{(j,s) \in \mu_{1,r_1^{\prime}} \mid s \in c_1^1 \text{, } \mu_{1,r_1^{\prime}}(s) \cap c_1^1 = \emptyset \text{, } \# \mu_{1,r_1^{\prime}}(s)=q_s \text{ and } F_s(j)>F_s(l) \text{ for all } l \in \mu_{1,r_1^{\prime}}(s), l \neq j \}$. Next $\mu_{1,1}^{\prime \prime} = \mu_{1,1}^{\prime} - \{(i,s) \mid i,s \in c_1^1 \text{ and } s \mapsto i\}$. Next $\mu_{1,1}^{\prime \prime \prime} = \mu_{1,1}^{\prime \prime} + \{(i,s) \mid i,s \in c_1^1 \text{ and } i \mapsto s\}$. We reach $\mu_{1,1}^{\prime \prime \prime}$ with $m_1^1 \subseteq \mu_{1,1}^{\prime \prime \prime}$. If $1 \neq L_1$, then go to Step 1.B.2. Otherwise, go to Step 1.End with $\mu_{1,L_1}^{\prime \prime \prime} = \mu_{1,1}^{\prime \prime \prime}$.

\item[Step 1.B.$k$.]($k>1$) If $m_1^k \subseteq \mu_{1,k-1}^{\prime \prime \prime}$ and $k \neq L_1$ then go to Step 1.B.$k$+1 with $\mu_{1,k}^{\prime \prime \prime} = \mu_{1,k-1}^{\prime \prime \prime}$. If $m_1^k \subseteq \mu_{1,k-1}^{\prime \prime \prime}$ and $k = L_1$ then go to Step 1.End with $\mu_{1,L_1}^{\prime \prime \prime} = \mu_{1,k-1}^{\prime \prime \prime}$. If $m_1^k \nsubseteq \mu_{1,k-1}^{\prime \prime \prime}$ then $\mu_{1,k}^{\prime} = \mu_{1,k-1}^{\prime \prime \prime} - \{(i,\mu_{1,k-1}^{\prime \prime \prime}(i)) \mid (i,\mu^C(i)) \in m_1^k \text{ and } \mu_{1,k-1}^{\prime \prime \prime}(i) \neq i\} + \{(i,s) \mid i,s \in c_1^k \text{ and } s \mapsto i\} - \{(j,s) \in \mu_{1,k-1}^{\prime \prime \prime} \mid s \in c_1^k \text{, } \mu_{1,k-1}^{\prime \prime \prime}(s) \cap c_1^k = \emptyset \text{, } \# \mu_{1,k-1}^{\prime \prime \prime}(s)=q_s \text{ and } F_s(j)>F_s(l) \text{ for all } l \in \mu_{1,k-1}^{\prime \prime \prime}(s), l \neq j \}$. Next $\mu_{1,k}^{\prime \prime} = \mu_{1,k}^{\prime} - \{(i,s) \mid i,s \in c_1^k \text{ and } s \mapsto i\}$. Next $\mu_{1,k}^{\prime \prime \prime} = \mu_{1,k}^{\prime \prime} + \{(i,s) \mid i,s \in c_1^k \text{ and } i \mapsto s\}$. We reach $\mu_{1,k}^{\prime \prime \prime}$ with $m_1^k \subseteq \mu_{1,k}^{\prime \prime \prime}$. If $k \neq L_1$, then go to Step 1.B.$k$+1. Otherwise, go to Step 1.End with $\mu_{1,L_1}^{\prime \prime \prime} = \mu_{1,k}^{\prime \prime \prime}$.

\item[Step 1.End.] We have reached $\mu_{1,L_1}^{\prime \prime \prime}$ with $\cup_{l=1}^{L_1} m_1^{l} = M_1^1 \subseteq \mu_{1,L_1}^{\prime \prime \prime}$ and $\cup_{r=1}^{r_1^{\prime}} m_1^{0,r}= M_1^0 \subseteq \mu_{1,L_1}^{\prime \prime \prime}$. That is, $M_1 \subseteq \mu_{1,L_1}^{\prime \prime \prime}$. If $\mu_{1,L_1}^{\prime \prime \prime} = \mu^C$ then the process ends. Otherwise, go to Step 2.A.1.

\item[Step 2.A.1.]If $m_2^{0,1} \subseteq \mu_{1,L_1}^{\prime \prime \prime}$ and $r_2^{\prime} \neq 1$ then go to Step 2.A.2 with $\mu_{2,1} = \mu_{1,L_1}^{\prime \prime \prime}$. If $m_2^{0,1} \subseteq \mu_{1,L_1}^{\prime \prime \prime}$, $r_2^{\prime}=1$ and $C_2 \neq \emptyset$ then go to Step 2.B.1 with $\mu_{2,r_2^{\prime}} = \mu_{1,L_1}^{\prime \prime \prime}$. If $m_2^{0,1} \subseteq \mu_{1,L_1}^{\prime \prime \prime}$, $r_2^{\prime}=1$ and $C_2 = \emptyset$ then go to Step 2.End with $\mu_{2,L_2}^{\prime \prime \prime} = \mu_{1,L_1}^{\prime \prime \prime}$. It is not excluded that $m_2^{0,1} = \emptyset$. Given the matches $M_1 \subseteq \mu_{1,L_1}^{\prime \prime \prime}$ remain fixed, let $\Lambda_{2,1}(s)= \#\{(i,s^{\prime}) \notin \mu_{1,L_1}^{\prime \prime \prime} \mid (i,s^{\prime}) \in m_2^{0,1} \text{ and } s^{\prime} = s\}$ be the number of students who are not yet matched to their preferred school $s$ in $\mu_{1,L_1}^{\prime \prime \prime}$ and are guaranteed admissions to their preferred school $s$. If $m_2^{0,1} \nsubseteq \mu_{1,L_1}^{\prime \prime \prime}$ then $\mu_{2,1} = \mu_{1,L_1}^{\prime \prime \prime} + \{(i,s) \notin \mu_{1,L_1}^{\prime \prime \prime} \mid (i,s) \in m_2^{0,1} \} - \{(j,s) \in \mu_{1,L_1}^{\prime \prime \prime} \mid \Lambda_j^s (\mu_{1,L_1}^{\prime \prime \prime}) < \Lambda_{2,1}(s) - q_s + \#\mu_{1,L_1}^{\prime \prime \prime}(s)\}$ where $\Lambda_j^s (\mu_{1,L_1}^{\prime \prime \prime}) = \#\{l \in I \mid (l,s) \in \mu_{1,L_1}^{\prime \prime \prime} \text{ and } F_s(l) > F_s(j)\}$ is the number of students who are matched to school $s$ in $\mu_{1,L_1}^{\prime \prime \prime}$ and have a lower priority than student $j$. We reach $\mu_{2,1}$ with $m_2^{0,1} \subseteq \mu_{2,1}$. If $r_2^{\prime} \neq 1$ then go to Step 2.A.2. If $r_2^{\prime} = 1$ and $C_2 \neq \emptyset$, then go to Step 2.B.1. Otherwise, go to Step 2.End with $\mu_{2,L_2}^{\prime \prime \prime} = \mu_{2,1}$.

\item[Step 2.A.$k$.]($k>1$) If $m_2^{0,k} \subseteq \mu_{2,k-1}$ and $r_2^{\prime} \neq k$ then go to Step 2.A.$k+1$ with $\mu_{2,k} = \mu_{2,k-1}$. If $m_2^{0,k} \subseteq \mu_{2,k-1}$, $r_2^{\prime}=k$ and $C_2 \neq \emptyset$ then go to Step 2.B.1 with $\mu_{2,r_2^{\prime}} = \mu_{2,k-1}$. If $m_2^{0,k} \subseteq \mu_{2,k-1}$, $r_2^{\prime}=k$ and $C_2 = \emptyset$ then go to Step 2.End with $\mu_{2,L_2}^{\prime \prime \prime} = \mu_{2,k-1}$. It is not excluded that $m_2^{0,k} = \emptyset$. Given the matches $M_1 \cup (\cup_{r=1}^{k-1} m_2^{0,r})  \subseteq \mu_{2,k-1}$ remain fixed, let $\Lambda_{2,k}(s)= \#\{(i,s^{\prime}) \notin \mu_{2,k-1} \mid (i,s^{\prime}) \in m_2^{0,k} \text{ and } s^{\prime} = s \}$ be the number of students who are not yet matched to their preferred school $s$ in $\mu_{2,k-1}$ and are guaranteed admissions to their preferred school $s$. If $m_2^{0,k} \nsubseteq \mu_{2,k-1}$ then $\mu_{2,k} = \mu_{2,k-1} + \{(i,s) \notin \mu_{2,k-1} \mid (i,s) \in m_2^{0,k} \} - \{(j,s) \in \mu_{2,k-1} \mid \Lambda_j^s (\mu_{2,k-1}) < \Lambda_{2,k}(s) - q_s + \#\mu_{2,k-1}(s)\}$ where $\Lambda_j^s (\mu_{2,k-1}) = \#\{l \in I \mid (l,s) \in \mu_{2,k-1} \text{ and } F_s(l) > F_s(j)\}$ is the number of students who are matched to school $s$ in $\mu_{2,k-1}$ and have a lower priority than student $j$. We reach $\mu_{2,k}$ with $m_2^{0,k} \subseteq \mu_{2,k}$. If $r_2^{\prime} \neq k$ then go to Step 2.A.$k+1$. If $r_2^{\prime} = k$ and $C_2 \neq \emptyset$, then go to Step 2.B.1. Otherwise, go to Step 2.End with $\mu_{2,L_2}^{\prime \prime \prime} = \mu_{2,k}$.

\item[Step 2.B.1.]If $m_2^1 \subseteq \mu_{2,r_2^{\prime}}$ and $1 \neq L_2$ then go to Step 2.B.2 with $\mu_{2,1}^{\prime \prime \prime} = \mu_{2,r_2^{\prime}}$. If $m_2^1 \subseteq \mu_{2,r_2^{\prime}}$ and $1 = L_2$ then go to Step 2.End with $\mu_{2,L_2}^{\prime \prime \prime} = \mu_{2,r_2^{\prime}}$. If $m_2^1 \nsubseteq \mu_{2,r_2^{\prime}}$ then $\mu_{2,1}^{\prime} = \mu_{2,r_2^{\prime}} - \{(i,\mu_{2,r_2^{\prime}}(i)) \mid (i,\mu^C(i)) \in m_2^1 \text{ and } \mu_{2,r_2^{\prime}}(i) \neq i \} + \{(i,s) \mid i,s \in c_2^1 \text{ and } s \mapsto i\} - \{(j,s) \in \mu_{2,r_2^{\prime}} \mid s \in c_2^1 \text{, } \mu_{2,r_2^{\prime}}(s) \cap c_2^1 = \emptyset \text{, } \# \mu_{2,r_2^{\prime}}(s)=q_s \text{ and } F_s(j)>F_s(l) \text{ for all } l \in \mu_{2,r_2^{\prime}}(s), l \neq j \}$. Next $\mu_{2,1}^{\prime \prime} = \mu_{2,1}^{\prime} - \{(i,s) \mid i,s \in c_2^1 \text{ and } s \mapsto i\}$. Next $\mu_{2,1}^{\prime \prime \prime} = \mu_{2,1}^{\prime \prime} + \{(i,s) \mid i,s \in c_2^1 \text{ and } i \mapsto s\}$. We reach $\mu_{2,1}^{\prime \prime \prime}$ with $m_2^1 \subseteq \mu_{2,1}^{\prime \prime \prime}$. If $1 \neq L_2$, then go to Step 2.B.2. Otherwise, go to Step 2.End with $\mu_{2,L_2}^{\prime \prime \prime} = \mu_{2,1}^{\prime \prime \prime}$.

\item[Step 2.B.$k$.]($k>1$) If $m_2^k \subseteq \mu_{2,k-1}^{\prime \prime \prime}$ and $k \neq L_2$ then go to Step 2.B.$k$+1 with $\mu_{2,k}^{\prime \prime \prime} = \mu_{2,k-1}^{\prime \prime \prime}$. If $m_2^k \subseteq \mu_{2,k-1}^{\prime \prime \prime}$ and $k = L_2$ then go to Step 2.End with $\mu_{2,L_2}^{\prime \prime \prime} = \mu_{2,k-1}^{\prime \prime \prime}$. If $m_2^k \nsubseteq \mu_{2,k-1}^{\prime \prime \prime}$ then $\mu_{2,k}^{\prime} = \mu_{2,k-1}^{\prime \prime \prime} - \{(i,\mu_{2,k-1}^{\prime \prime \prime}(i)) \mid (i,\mu^C(i)) \in m_2^k \text{ and } \mu_{2,k-1}^{\prime \prime \prime}(i) \neq i\} + \{(i,s) \mid i,s \in c_2^k \text{ and } s \mapsto i\} - \{(j,s) \in \mu_{2,k-1}^{\prime \prime \prime} \mid s \in c_2^k \text{, } \mu_{2,k-1}^{\prime \prime \prime}(s) \cap c_2^k = \emptyset \text{, } \# \mu_{2,k-1}^{\prime \prime \prime}(s)=q_s \text{ and } F_s(j)>F_s(l) \text{ for all } l \in \mu_{2,k-1}^{\prime \prime \prime}(s), l \neq j \}$. Next $\mu_{2,k}^{\prime \prime} = \mu_{2,k}^{\prime} - \{(i,s) \mid i,s \in c_2^k \text{ and } s \mapsto i\}$. Next $\mu_{2,k}^{\prime \prime \prime} = \mu_{2,k}^{\prime \prime} + \{(i,s) \mid i,s \in c_2^k \text{ and } i \mapsto s\}$. We reach $\mu_{2,k}^{\prime \prime \prime}$ with $m_2^k \subseteq \mu_{2,k}^{\prime \prime \prime}$. If $k \neq L_2$, then go to Step 2.B.$k$+1. Otherwise, go to Step 2.End with $\mu_{2,L_2}^{\prime \prime \prime} = \mu_{2,k}^{\prime \prime \prime}$.

\item[Step 2.End.] We have reached $\mu_{2,L_2}^{\prime \prime \prime}$ with $\cup_{l=1}^{L_2} m_2^{l} = M_2^1 \subseteq \mu_{2,L_2}^{\prime \prime \prime}$, $\cup_{r=1}^{r_2^{\prime}} m_2^{0,r}= M_2^0 \subseteq \mu_{2,L_2}^{\prime \prime \prime}$ and $M_1 \subseteq \mu_{2,L_2}^{\prime \prime \prime}$. That is, $M_1 \cup M_2 \subseteq \mu_{2,L_2}^{\prime \prime \prime}$. If $\mu_{2,L_2}^{\prime \prime \prime} = \mu^C$ then the process ends. Otherwise, go to Step 3.A.1.

\item[End.]The process goes on until we reach $\mu_{\bar{k},L_{\bar{k}}}^{\prime \prime \prime} = \cup_{k=1}^{\bar{k}} M_k = \mu^C$.

\end{itemize}
\begin{flushright}
\qedsymbol
\end{flushright}

\subsection{Equitable Top Trading Cycles Algorithm}\label{AP3} 

The Equitable Top Trading Cycles mechanism (Hakimov and Kesten, 2018) finds a matching by means of the following Equitable Top Trading Cycles algorithm (ETTC).

\begin{itemize}
\item[Step $1$.] Set $q_s^{1} = q_s$ for all $s \in S$ where $q_s^{1}$ is the initial capacity of school $s$ at Step $1$. 

\item[$1$.A.] In the inheritance round, since all seats are available to inherit, students are assigned seats according to the priority orders $F$ to form student-school pairs. Let $\Phi^1(s,i) = \{j \in I \mid F_s(j) < F_s(i)\}$ be the set of students who have higher priority than student $i$ for school $s$ in Step 1. Let $\mathcal{IS}_{1}=\{(i,s) \in I \times S \mid \# \Phi^1(s,i) \leq q^{1}_{s}\}$ be the set of student-school pairs formed by assigning students one-by-one to the schools while respecting their capacities. In other words, $\mathcal{IS}_{1}$ consists of student-school pairs such that each school $s$ pairs with $q_{s}$ highest priority students.

\item[$1$.B.] Each student-school pair $(i,s)\in\mathcal{IS}_{1}$ points to the student-school pair $(i^{\prime},s^{\prime})\in\mathcal{IS}_{1}$ such that (1) $s^{\prime}$ is the best choice of student $i$ in $P_{i}$, and (2) student $i^{\prime}$ has the highest priority in $F_{s}$ among students who are assigned a seat at $s^{\prime}$, i.e. $(i,s)\mapsto (i^{\prime},s^{\prime})$ such that $F_{s}(i^{\prime}) < F_{s}(l)$ for any other $(l,s^{\prime})\in\mathcal{IS}_{1}$.\footnote{Note that if $F_{s}(i)=1$, then all pairs $(i,s^{\prime})\in\mathcal{IS}_{1}$ point to $(i,s)$.} Since there is a finite number of students and schools, there is at least one cycle. Let $C_{1}=\{c_{1}^{1},c^{2}_{1},\ldots, c_{1}^{L_{1}}\}$ be the set of cycles in Step $1$.B where $L_1 \ge 1$ is the number of cycles in Step $1$.B. 

\item[$1$.C.] If student $i$ appears in the same cycle or in different cycles with different schools, then she is assigned a seat at her top choice among those schools. That is, for each $i\in I$ such that there exists $(i,s)\mapsto (i^{\prime},s^{\prime})$ in $c_{1}^{l}$ and $(i,\hat{s}) \mapsto (i^{\prime \prime},s^{\prime \prime})$ in $c_{1}^{l^{\prime}}$, possibly $c_{1}^{l}=c_{1}^{l^{\prime}}$, $m_{1}(i)=s^{\prime}$ such that $P_{i}(s^{\prime}) > P_{i}(s^{\prime \prime})$. Moreover, the seats at all other schools than her top choice she points to in those cycles remain to be inherited in Step $2$.A. For all other students, they are matched with the school that is in the student-school pair they point to at a cycle $c^{l}_{1}$, i.e. if $(i,s) \mapsto (i^{\prime},s^{\prime})$ (possibly $s = s^{\prime}$) in $c^{l}_{1}$, then $m_{1}(i) = s^{\prime}$. Finally, if there is a student-school pair participating in a cycle, $(i,s) \in c^{l}_{1}$, and another student-school pair with the same student and a different school not participating at any cycle, $(i,s^{\prime}) \notin c_{1}^{l^{\prime}}$, $c_{1}^{l^{\prime}} \in C_1$, then this seat at school $s^{\prime}$ remains to be inherited in Step $2$.A. Let $I_1 = \{i \in I \mid (i,s) \in c_1^{l} \text{, } c_1^{l} \in C_1\}$ be the set of students involved in a cycle in Step 1. Let $M_{1} = \cup_{i \in I_1}(i,m_{1}(i))$ be all the matches formed between students and schools in Step 1. Let $\widehat{I}_1 =I \setminus I_{1}$ be the set of students who have not been assigned a seat at the end of Step 1. If $\widehat{I}_1 \neq \emptyset$, then go to Step $2$.A. Otherwise, go to End.

\item[Step $k \geq 2$.] At the beginning of Step $k$, the remaining capacity of school $s$ is $q_s^{k}$ and the set of remaining students is $\widehat{I}_{k-1}$.

\item[$k$.A.] In the inheritance round, for each school $s$ such that (1) there are seats at $s$ remained from Step $k-1$.C to be inherited, and (2) no student-school pair was assigned at Step $k-1$ and hence remaining seats $q_{s}^{k}$ are assigned to the remaining students $\widehat{I}_{k-1}$ according to the priority orders $F$ to form student-school pairs. Let $\Phi^k(s,i) = \{j \in \widehat{I}_{k-1} \mid F_s(j) < F_s(i)\}$ be the set of students who have higher priority than student $i$ for school $s$ in Step $k$. Let $\mathcal{IS}_{k} = \{(i,s) \in \widehat{I}_{k-1} \times S \mid \# \Phi^k(s,i) \leq q^{k}_{s} \}$ be the set of student-school pairs formed by assigning students one-by-one to the schools while respecting their capacities. In other words, $\mathcal{IS}_{k}$ consists of student-school pairs such that each school $s$ pairs with $q_{s}^{k}$ highest priority students.

\item[$k$.B.] Each student-school pair $(i,s) \in \mathcal{IS}_{k}$ points to the student-school pair $(i^{\prime},s^{\prime}) \in \mathcal{IS}_{k}$ such that (1) $s^{\prime}$ is the best choice of student $i$ in $P_{i}$, and (2) student $i^{\prime}$ has the highest priority in $F_{s}$ among students that are assigned a seat at $s^{\prime}$, i.e. $(i,s) \mapsto (i^{\prime},s^{\prime})$ such that $F_{s}(i^{\prime}) < F_{s}(l)$ for any other $(l,s^{\prime}) \in \mathcal{IS}_{k}$. Since there is a finite number of students and schools, there is at least one cycle. Let $C_{k}=\{c_{k}^{1},c^{2}_{k}, \ldots, c_{k}^{L_{k}}\}$ be the set of cycles in Step $k$.B where $L_k \ge 1$ is the number of cycles in Step $k$.B. 

\item[$k$.C.] If student $i$ appears in the same cycle or in different cycles with different schools, then she is assigned a seat at her top choice among those schools. That is, for each $i\in \widehat{I}_{k-1}$ such that there exists $(i,s) \mapsto (i^{\prime},s^{\prime})$ in $c_{k}^{l}$ and $(i,\hat{s}) \mapsto (i^{\prime \prime},s^{\prime \prime})$ in $c_{k}^{l^{\prime}}$, possibly $c_{k}^{l}=c_{k}^{l^{\prime}}$, $m_{k}(i)=s^{\prime}$ such that $P_{i}(s^{\prime}) > P_{i}(s^{\prime \prime})$. Moreover, the seats at all other schools than her top choice she points to in those cycles remain to be inherited in Step $k$.A. For all other students, they are matched with the school that is in the student-school pair they point to at a cycle $c^{l}_{k}$, i.e. if $(i,s)\mapsto (i^{\prime},s^{\prime})$ (possibly $s = s^{\prime}$) in $c^{l}_{k}$, then $m_{k}(i) = s^{\prime}$. Finally, if there is a student-school pair participating in a cycle, $(i,s) \in c^{l}_{k}$, and another student-school pair with the same student and a different school not participating at any cycle, $(i,s^{\prime}) \notin c_{k}^{l^{\prime}}$, $c_{k}^{l^{\prime}} \in C_k$, then this seat at school $s^{\prime}$ remains to be inherited in Step $k+1$.A. Let $I_k = \{i \in \widehat{I}_{k-1} \mid (i,s) \in c_k^{l} \text{, } c_k^{l} \in C_k\}$ be the set of students involved in a cycle in Step $k$. Let $M_{k} = \cup_{i \in I_k}(i,m_{k}(i))$ be all the matches formed between students and schools in Step $k$. Let $\widehat{I}_k = \widehat{I}_{k-1} \setminus I_{k}$ be the set of students who have not been assigned a seat at the end of Step $k$. If $\widehat{I}_k \neq \emptyset$, then go to Step $k+1$.A. Otherwise, go to End.

\item[End] The algorithm stops when all students have been removed. Let $\bar{k}$ be the step at which the algorithm stops. Let $\mu^{E}$ denote the matching obtained from the ETTC algorithm and it is given by $\mu^{E}= \cup_{k=1}^{\bar{k}}M_{k}$.
 
\end{itemize}

\noindent \textbf{Proof of Theorem \ref{THETTC}}

Since $\{\mu^E\}$ is a singleton set, internal stability (IS) is satisfied. (ES) Take any matching $\mu \neq \mu^E$, we need to show that $\phi(\mu) \ni \mu^E$. We build in steps a farsighted improving path from $\mu$ to $\mu^E$.
\begin{itemize}

\item[Step 1.1.] If $(i,\mu^E(i)) \in \mu$ for all $i \in \{j \in I \mid (j,s) \in c_1^1\}$ and $1 \neq L_1$ then go to Step 1.2 with $\mu_{1,1}^{\prime \prime \prime } = \mu$. If $(i,\mu^E(i)) \in \mu$ for all $i \in \{j \in I \mid (j,s) \in c_1^1\}$ and $1 = L_1$ then go to Step 1.End with $\mu_{1,L_1}^{\prime \prime \prime } = \mu$. Let $c_1^1(i)=\{(i,s^l)\}_{l=1}^{\tau(i,c_1^1)}$ such that $(i,s^l) \in c_1^1$ and $s^l = s_{o_l} \neq s^{l+1} = s_{o_{l+1}}$ with $o_{l} < o_{l+1}$ for $l=1,...,\tau(i,c_1^1)-1$. That is, $c_1^1(i)$ is an ordered set of the pairs involving student $i$ in cycle $c_1^1$ where $\tau(i,c_1^1)=\#\{(j,s) \in c_1^1 \mid j=i\}$ is the number of distinct pairs involving student $i$ in cycle $c_1^1$. Let $\Lambda_{1,1}(s)= \#\{(i,s^{\prime}) \notin \mu \mid (i,s^{\prime}) = (i,s^1) \text{ with } (i,s^1) \in c_1^1(i) \text{ and } s^{\prime} = s\}$ be the number of students who are not yet matched in $\mu$ to school $s$ that ranks them among the first $q_s$ positions and is ranked first in their ordered set. If $(i,\mu^E(i)) \notin \mu$ for some $i \in \{j \in I \mid (j,s) \in c_1^1\}$ then $\mu_{1,1}^{\prime} = \mu - \{(i,\mu(i)) \mid (i,s) \in c_1^1 \text{ and } \mu(i) \neq i\} + \{(i,s) \mid (i,s) = (i,s^1) \text{ with } (i,s^1) \in c_1^1(i)  \} - \{(j,s) \in \mu \mid \Lambda_j^s (\mu) < \Lambda_{1,1}(s) - q_s + \#\mu(s)\}$ where $\Lambda_j^s (\mu) = \#\{l \in I \mid (l,s) \in \mu \text{ and } F_s(l) > F_s(j)\}$ is the number of students who are matched to school $s$ in $\mu$ and have a lower priority than student $j$. Next, if $(i,\mu_{1,1}^{\prime}(i))=(i,\mu^E(i))$ and $c_1^1=\{(i,\mu^E(i))\}$ then $\mu_{1,1}^{\prime \prime} = \mu_{1,1}^{\prime}$. Otherwise, $\mu_{1,1}^{\prime \prime} = \mu_{1,1}^{\prime} - \{(i,s) \mid (i,s) = (i,s^1) \text{ with } (i,s^1) \in c_1^1(i)  \}$ so that all students involved in $c_1^1$ are unmatched. If $\tau(i,c_1^1) = 1$ for all $i \in \{j \in I \mid (j,s) \in c_1^1\}$ and $1 \neq L_1$, then go to Step 1.2. If $\tau(i,c_1^1) = 1$ for all $i \in \{j \in I \mid (j,s) \in c_1^1\}$ and $1 = L_1$, then go to Step 1.End with $\mu_{1,L_1}^{\prime \prime \prime} = \mu_{1,1}^{\prime \prime }$. If $\tau(i,c_1^1) \neq 1$ for some $i \in \{j \in I \mid (j,s) \in c_1^1\}$ then go to Step 1.1.A.

\item[Step 1.1.A] Take $i \in \{j \in I \mid (j,s) \in c_1^1\}$ such that $\tau(i,c_1^1) \neq 1$. If $\#\mu_{1,1}^{\prime \prime}(s^2) < q_{s^2}$, then $\mu_{1,1}^{i2} = \mu_{1,1}^{\prime \prime} + (i,s^2)$ with $(i,s^2) \in c_1^1(i)$. If $\#\mu_{1,1}^{\prime \prime}(s^2) = q_{s^2}$, then $\mu_{1,1}^{i2} = \mu_{1,1}^{\prime \prime} + (i,s^2) - (j,s^2)$ with $(i,s^2) \in c_1^1(i)$, $(j,s^2) \in \mu_{1,1}^{\prime \prime}$ and $F_{s^2}(j)>F_{s^2}(l)$ for all $l\in \mu_{1,1}^{\prime \prime}(s^2)$, $l \neq j$. Next, student $i$ leaves school $s^2$ to become unmatched and guaranteeing a free slot at school $s^2$. We reach $\mu_{1,1}^{i2^{\prime}} = \mu_{1,1}^{i2} - (i,s^2)$. Next, if $\tau(i,c_1^1) \neq 2$ and $\#\mu_{1,1}^{\prime \prime}(s^3) < q_{s^3}$, then $\mu_{1,1}^{i3} = \mu_{1,1}^{i2^{\prime}} + (i,s^3)$ with $(i,s^3) \in c_1^1(i)$. If $\tau(i,c_1^1) \neq 2$ and $\#\mu_{1,1}^{\prime \prime}(s^3) = q_{s^3}$, then $\mu_{1,1}^{i3} = \mu_{1,1}^{i2^{\prime}} + (i,s^3) - (j,s^3)$ with $(i,s^3) \in c_1^1(i)$, $(j,s^3) \in \mu_{1,1}^{\prime \prime}$ and $F_{s^3}(j)>F_{s^3}(l)$ for all $l\in \mu_{1,1}^{\prime \prime}(s^3)$, $l \neq j$. Next, student $i$ leaves school $s^3$ to become unmatched and guaranteeing a free slot at school $s^3$. We reach $\mu_{1,1}^{i3^{\prime}} = \mu_{1,1}^{i3} - (i,s^3)$. We repeat this process until we reach $\mu_{1,1}^{i\tau(i,c_1^1)^{\prime}} = \mu_{1,1}^{i\tau(i,c_1^1)} - (i,s^{\tau(i,c_1^1)})$.

We repeat the process of Step 1.1.A with each student $i \in \{j \in I \mid (j,s) \in c_1^1 \}$ such that $\tau(i,c_1^1) \neq 1$ to reach in the end the matching $\mu_{1,1}^{\prime \prime \prime}$ where all students involved in $c_1^1$ are unmatched and each school $s$ involved in $c_1^1$ has $\#\{(i,s^\prime) \in c_1^1 \mid s^\prime = s \}$ free slots.

\item[Step 1.$k$.]($k>1$) If $(i,\mu^E(i)) \in \mu_{1,k-1}^{\prime \prime \prime}$ for all $i \in \{j \in I \mid (j,s) \in c_1^k\}$ and $k \neq L_1$ then go to Step 1.k+1 with $\mu_{1,k}^{\prime \prime \prime } = \mu_{1,k-1}^{\prime \prime \prime}$. If $(i,\mu^E(i)) \in \mu_{1,k-1}^{\prime \prime \prime}$ for all $i \in \{j \in I \mid (j,s) \in c_1^k\}$ and $k = L_1$ then go to Step 1.End with $\mu_{1,L_1}^{\prime \prime \prime } = \mu_{1,k-1}^{\prime \prime \prime}$. Let $c_1^k(i)=\{(i,s^l)\}_{l=1}^{\tau(i,c_1^k)}$ such that $(i,s^l) \in c_1^k$ and $s^l = s_{o_l} \neq s^{l+1} = s_{o_{l+1}}$ with $o_{l} < o_{l+1}$ for $l=1,...,\tau(i,c_1^k)-1$. That is, $c_1^k(i)$ is an ordered set of the pairs involving student $i$ in cycle $c_1^k$ where $\tau(i,c_1^k)=\#\{(j,s) \in c_1^k \mid j=i\}$ is the number of distinct pairs involving student $i$ in cycle $c_1^k$. Let $\Lambda_{1,k}(s)= \#\{(i,s^{\prime}) \notin \mu_{1,k-1}^{\prime \prime \prime} \mid (i,s^{\prime}) = (i,s^1) \text{ with } (i,s^1) \in c_1^k(i) \text{ and } s^{\prime} = s\}$ be the number of students who are not yet matched in $\mu_{1,k-1}^{\prime \prime \prime}$ to school $s$ that ranks them among the first $q_s$ positions and is ranked first in their ordered set. If $(i,\mu^E(i)) \notin \mu_{1,k-1}^{\prime \prime \prime}$ for some $i \in \{j \in I \mid (j,s) \in c_1^k\}$ then $\mu_{1,k}^{\prime} = \mu_{1,k-1}^{\prime \prime \prime} - \{(i,\mu_{1,k-1}^{\prime \prime \prime}(i)) \mid (i,s) \in c_1^k \text{ and } \mu_{1,k-1}^{\prime \prime \prime}(i) \neq i\} + \{(i,s) \mid (i,s) = (i,s^1) \text{ with } (i,s^1) \in c_1^k(i)  \} - \{(j,s) \in \mu_{1,k-1}^{\prime \prime \prime} \mid \Lambda_j^s (\mu_{1,k-1}^{\prime \prime \prime}) < \Lambda_{1,k}(s) - q_s + \#\mu_{1,k-1}^{\prime \prime \prime}(s)\}$ where $\Lambda_j^s (\mu_{1,k-1}^{\prime \prime \prime}) = \#\{l \in I \mid (l,s) \in \mu_{1,k-1}^{\prime \prime \prime} \text{ and } F_s(l) > F_s(j)\}$ is the number of students who are matched to school $s$ in $\mu_{1,k-1}^{\prime \prime \prime}$ and have a lower priority than student $j$. Next, if $(i,\mu_{1,k}^{\prime}(i))=(i,\mu^E(i))$ and $c_1^k=\{(i,\mu^E(i))\}$ then $\mu_{1,k}^{\prime \prime} = \mu_{1,k}^{\prime}$. Otherwise, $\mu_{1,k}^{\prime \prime} = \mu_{1,k}^{\prime} - \{(i,s) \mid (i,s) = (i,s^1) \text{ with } (i,s^1) \in c_1^k(i)  \}$ so that all students involved in $c_1^k$ are unmatched. If $\tau(i,c_1^k) = 1$ for all $i \in \{j \in I \mid (j,s) \in c_1^k\}$ and $k \neq L_1$, then go to Step 1.k+1. If $\tau(i,c_1^k) = 1$ for all $i \in \{j \in I \mid (j,s) \in c_1^k\}$ and $k = L_1$, then go to Step 1.End with $\mu_{1,L_1}^{\prime \prime \prime} = \mu_{1,k}^{\prime \prime }$. If $\tau(i,c_1^k) \neq 1$ for some $i \in \{j \in I \mid (j,s) \in c_1^k\}$ then go to Step 1.k.A.

\item[Step 1.k.A] Take $i \in \{j \in I \mid (j,s) \in c_1^k\}$ such that $\tau(i,c_1^k) \neq 1$. If $\#\mu_{1,k}^{\prime \prime}(s^2) < q_{s^2}$, then $\mu_{1,k}^{i2} = \mu_{1,k}^{\prime \prime} + (i,s^2)$ with $(i,s^2) \in c_1^k(i)$. If $\#\mu_{1,k}^{\prime \prime}(s^2) = q_{s^2}$, then $\mu_{1,k}^{i2} = \mu_{1,k}^{\prime \prime} + (i,s^2) - (j,s^2)$ with $(i,s^2) \in c_1^k(i)$, $(j,s^2) \in \mu_{1,k}^{\prime \prime}$ and $F_{s^2}(j)>F_{s^2}(l)$ for all $l\in \mu_{1,k}^{\prime \prime}(s^2)$, $l \neq j$. Next, student $i$ leaves school $s^2$ to become unmatched and guaranteeing a free slot at school $s^2$. We reach $\mu_{1,k}^{i2^{\prime}} = \mu_{1,k}^{i2} - (i,s^2)$. Next, if $\tau(i,c_1^k) \neq 2$ and $\#\mu_{1,k}^{\prime \prime}(s^3) < q_{s^3}$, then $\mu_{1,k}^{i3} = \mu_{1,k}^{i2^{\prime}} + (i,s^3)$ with $(i,s^3) \in c_1^k(i)$. If $\tau(i,c_1^k) \neq 2$ and $\#\mu_{1,k}^{\prime \prime}(s^3) = q_{s^3}$, then $\mu_{1,k}^{i3} = \mu_{1,k}^{i2^{\prime}} + (i,s^3) - (j,s^3)$ with $(i,s^3) \in c_1^k(i)$, $(j,s^3) \in \mu_{1,k}^{\prime \prime}$ and $F_{s^3}(j)>F_{s^3}(l)$ for all $l\in \mu_{1,k}^{\prime \prime}(s^3)$, $l \neq j$. Next, student $i$ leaves school $s^3$ to become unmatched and guaranteeing a free slot at school $s^3$. We reach $\mu_{1,k}^{i3^{\prime}} = \mu_{1,k}^{i3} - (i,s^3)$. We repeat this process until we reach $\mu_{1,k}^{i\tau(i,c_1^k)^{\prime}} = \mu_{1,k}^{i\tau(i,c_1^k)} - (i,s^{\tau(i,c_1^k)})$.

We repeat the process of Step 1.k.A with each student $i \in \{j \in I \mid (j,s) \in c_1^k\}$ such that $\tau(i,c_1^k) \neq 1$ to reach in the end the matching $\mu_{1,k}^{\prime \prime \prime}$ where all students involved in $c_1^k$ are unmatched and each school $s$ involved in $c_1^k$ has $\#\{(i,s^\prime) \in c_1^k \mid s^\prime = s \}$ free slots. If $k \neq L_1$, then go to Step 1.$k$+1. Otherwise, go to Step 1.End with $\mu_{1,L_1}^{\prime \prime \prime} = \mu_{1,k}^{\prime \prime \prime}$.

\item[Step 1.End.] We have reached $\mu_{1,L_1}^{\prime \prime \prime}$ where each student $i$ involved in $C_1$ is either matched to $\mu^E(i)$ or unmatched and each school $s$ involved in $C_1$ has $\#\{(i,s^\prime) \in \cup_{l=1}^{L_1}c_1^l \mid s^\prime = s \text{ and } \mu_{1,L_1}^{\prime \prime \prime}(i) \neq \mu^E(i)\}$ free slots. Next, those unmatched students join the school they point to in $C_1$ to form the matching $\widetilde{\mu}_1 = \mu_{1,L_1}^{\prime \prime \prime} + \{(i,s) \in M_1 \mid (i,s) \notin \mu_{1,L_1}^{\prime \prime \prime} \}$. If $\widetilde{\mu}_1 = \mu^E$ then the process ends. Otherwise, go to Step 2.1.

\item[Step 2.1.] If $(i,\mu^E(i)) \in \widetilde{\mu}_1$ for all $i \in \{j \in \widehat{I}_{1} \mid (j,s) \in c_2^1\}$ and $1 \neq L_2$ then go to Step 2.2 with $\mu_{2,1}^{\prime \prime \prime } = \widetilde{\mu}_1$. If $(i,\mu^E(i)) \in \widetilde{\mu}_1$ for all $i \in \{j \in \widehat{I}_{1} \mid (j,s) \in c_2^1\}$ and $1 = L_2$ then go to Step 2.End with $\mu_{2,L_2}^{\prime \prime \prime } = \widetilde{\mu}_1$. Let $c_2^1(i)=\{(i,s^l)\}_{l=1}^{\tau(i,c_2^1)}$ such that $(i,s^l) \in c_2^1$ and $s^l = s_{o_l} \neq s^{l+1} = s_{o_{l+1}}$ with $o_{l} < o_{l+1}$ for $l=1,...,\tau(i,c_2^1)-1$. That is, $c_2^1(i)$ is an ordered set of the pairs involving student $i$ in cycle $c_2^1$ where $\tau(i,c_2^1)=\#\{(j,s) \in c_2^1 \mid j=i\}$ is the number of distinct pairs involving student $i$ in cycle $c_2^1$. Let $\Lambda_{2,1}(s)= \#\{(i,s^{\prime}) \notin \widetilde{\mu}_1 \mid (i,s^{\prime}) = (i,s^1) \text{ with } (i,s^1) \in c_2^1(i) \text{ and } s^{\prime} = s\}$ be the number of students who are not yet matched in $\widetilde{\mu}_1$ to school $s$ that ranks them among the first $q_s^2$ positions and is ranked first in their ordered set. If $(i,\mu^E(i)) \notin \widetilde{\mu}_1$ for some $i \in \{j \in \widehat{I}_{1} \mid (j,s) \in c_2^1\}$ then $\mu_{2,1}^{\prime} = \widetilde{\mu}_1 - \{(i,\widetilde{\mu}_1(i)) \mid (i,s) \in c_2^1 \text{ and } \widetilde{\mu}_1(i) \neq i\} + \{(i,s) \mid (i,s) = (i,s^1) \text{ with } (i,s^1) \in c_2^1(i)  \} - \{(j,s) \in \widetilde{\mu}_1 \mid \Lambda_j^s (\widetilde{\mu}_1) < \Lambda_{2,1}(s) - q_s + \#\widetilde{\mu}_1(s)\}$ where $\Lambda_j^s (\widetilde{\mu}_1) = \#\{l \in I \mid (l,s) \in \widetilde{\mu}_1 \text{ and } F_s(l) > F_s(j)\}$ is the number of students who are matched to school $s$ in $\widetilde{\mu}_1$ and have a lower priority than student $j$. Next, if $(i,\mu_{2,1}^{\prime}(i))=(i,\mu^E(i))$ and $c_2^1=\{(i,\mu^E(i))\}$ then $\mu_{2,1}^{\prime \prime} = \mu_{2,1}^{\prime}$. Otherwise,  $\mu_{2,1}^{\prime \prime} = \mu_{2,1}^{\prime} - \{(i,s) \mid (i,s) = (i,s^1) \text{ with } (i,s^1) \in c_2^1(i)  \}$ so that all students involved in $c_2^1$ are unmatched. If $\tau(i,c_2^1) = 1$ for all $i \in \{j \in \widehat{I}_{1} \mid (j,s) \in c_2^1\}$ and $1 \neq L_2$, then go to Step 2.2. If $\tau(i,c_2^1) = 1$ for all $i \in \{j \in \widehat{I}_{1} \mid (j,s) \in c_2^1\}$ and $1 = L_2$, then go to Step 2.End with $\mu_{2,L_2}^{\prime \prime \prime} = \mu_{2,1}^{\prime \prime }$. If $\tau(i,c_2^1) \neq 1$ for some $i \in \{j \in \widehat{I}_{1} \mid (j,s) \in c_2^1\}$ then go to Step 2.1.A.

\item[Step 2.1.A] Take $i \in \{j \in \widehat{I}_{1} \mid (j,s) \in c_2^1\}$ such that $\tau(i,c_2^1) \neq 1$. If $\#\mu_{2,1}^{\prime \prime}(s^2) < q_{s^2}$, then $\mu_{2,1}^{i2} = \mu_{2,1}^{\prime \prime} + (i,s^2)$ with $(i,s^2) \in c_2^1(i)$. If $\#\mu_{2,1}^{\prime \prime}(s^2) = q_{s^2}$, then $\mu_{2,1}^{i2} = \mu_{2,1}^{\prime \prime} + (i,s^2) - (j,s^2)$ with $(i,s^2) \in c_2^1(i)$, $(j,s^2) \in \mu_{2,1}^{\prime \prime}$ and $F_{s^2}(j)>F_{s^2}(l)$ for all $l\in \mu_{2,1}^{\prime \prime}(s^2)$, $l \neq j$. Next, student $i$ leaves school $s^2$ to become unmatched and guaranteeing a free slot at school $s^2$. We reach $\mu_{2,1}^{i2^{\prime}} = \mu_{2,1}^{i2} - (i,s^2)$. Next, if $\tau(i,c_2^1) \neq 2$ and $\#\mu_{2,1}^{\prime \prime}(s^3) < q_{s^3}$, then $\mu_{2,1}^{i3} = \mu_{2,1}^{i2^{\prime}} + (i,s^3)$ with $(i,s^3) \in c_2^1(i)$. If $\tau(i,c_2^1) \neq 2$ and $\#\mu_{2,1}^{\prime \prime}(s^3) = q_{s^3}$, then $\mu_{2,1}^{i3} = \mu_{2,1}^{i2^{\prime}} + (i,s^3) - (j,s^3)$ with $(i,s^3) \in c_2^1(i)$, $(j,s^3) \in \mu_{2,1}^{\prime \prime}$ and $F_{s^3}(j)>F_{s^3}(l)$ for all $l\in \mu_{2,1}^{\prime \prime}(s^3)$, $l \neq j$. Next, student $i$ leaves school $s^3$ to become unmatched and guaranteeing a free slot at school $s^3$. We reach $\mu_{2,1}^{i3^{\prime}} = \mu_{2,1}^{i3} - (i,s^3)$. We repeat this process until we reach $\mu_{2,1}^{i\tau(i,c_2^1)^{\prime}} = \mu_{2,1}^{i\tau(i,c_2^1)} - (i,s^{\tau(i,c_2^1)})$.

We repeat the process of Step 2.1.A with each student $i \in \{j \in \widehat{I}_{1} \mid (j,s) \in c_2^1\}$ such that $\tau(i,c_2^1) \neq 1$ to reach in the end the matching $\mu_{2,1}^{\prime \prime \prime}$ where all students involved in $c_2^1$ are unmatched and each school $s$ involved in $c_2^1$ has $\#\{(i,s^\prime) \in c_2^1 \mid s^\prime = s \}$ free slots.

\item[Step 2.$k$.]($k>1$) If $(i,\mu^E(i)) \in \mu_{2,k-1}^{\prime \prime \prime}$ for all $i \in \{j \in \widehat{I}_{1} \mid (j,s) \in c_2^k\}$ and $k \neq L_2$ then go to Step 2.k+1 with $\mu_{2,k}^{\prime \prime \prime } = \mu_{2,k-1}^{\prime \prime \prime}$. If $(i,\mu^E(i)) \in \mu_{2,k-1}^{\prime \prime \prime}$ for all $i \in \{j \in \widehat{I}_{1} \mid (j,s) \in c_2^k\}$ and $k = L_2$ then go to Step 2.End with $\mu_{2,L_2}^{\prime \prime \prime } = \mu_{2,k-1}^{\prime \prime \prime}$. Let $c_2^k(i)=\{(i,s^l)\}_{l=1}^{\tau(i,c_2^k)}$ such that $(i,s^l) \in c_2^k$ and $s^l = s_{o_l} \neq s^{l+1} = s_{o_{l+1}}$ with $o_{l} < o_{l+1}$ for $l=1,...,\tau(i,c_2^k)-1$. That is, $c_2^k(i)$ is an ordered set of the pairs involving student $i$ in cycle $c_2^k$ where $\tau(i,c_2^k)=\#\{(j,s) \in c_2^k \mid j=i\}$ is the number of distinct pairs involving student $i$ in cycle $c_2^k$. Let $\Lambda_{2,k}(s)= \#\{(i,s^{\prime}) \notin \mu_{2,k-1}^{\prime \prime \prime} \mid (i,s^{\prime}) = (i,s^1) \text{ with } (i,s^1) \in c_2^k(i) \text{ and } s^{\prime} = s\}$ be the number of students who are not yet matched in $\mu_{2,k-1}^{\prime \prime \prime}$ to school $s$ that ranks them among the first $q_s^2$ positions and is ranked first in their ordered set. If $(i,\mu^E(i)) \notin \mu_{2,k-1}^{\prime \prime \prime}$ for some $i \in \{j \in \widehat{I}_{1} \mid (j,s) \in c_2^k\}$ then $\mu_{2,k}^{\prime} = \mu_{2,k-1}^{\prime \prime \prime} - \{(i,\mu_{2,k-1}^{\prime \prime \prime}(i)) \mid (i,s) \in c_2^k \text{ and } \mu_{2,k-1}^{\prime \prime \prime}(i) \neq i\} + \{(i,s) \mid (i,s) = (i,s^1) \text{ with } (i,s^1) \in c_2^k(i)  \} - \{(j,s) \in \mu_{2,k-1}^{\prime \prime \prime} \mid \Lambda_j^s (\mu_{2,k-1}^{\prime \prime \prime}) < \Lambda_{2,k}(s) - q_s + \#\mu_{2,k-1}^{\prime \prime \prime}(s)\}$ where $\Lambda_j^s (\mu_{2,k-1}^{\prime \prime \prime}) = \#\{l \in I \mid (l,s) \in \mu_{2,k-1}^{\prime \prime \prime} \text{ and } F_s(l) > F_s(j)\}$ is the number of students who are matched to school $s$ in $\mu_{2,k-1}^{\prime \prime \prime}$ and have a lower priority than student $j$. Next, if $(i,\mu_{2,k}^{\prime}(i))=(i,\mu^E(i))$ and $c_2^k=\{(i,\mu^E(i))\}$ then $\mu_{2,k}^{\prime \prime} = \mu_{2,k}^{\prime}$. Otherwise,  $\mu_{2,k}^{\prime \prime} = \mu_{2,k}^{\prime} - \{(i,s) \mid (i,s) = (i,s^1) \text{ with } (i,s^1) \in c_2^k(i)  \}$ so that all students involved in $c_2^k$ are unmatched. If $\tau(i,c_2^k) = 1$ for all $i \in \{j \in \widehat{I}_{1} \mid (j,s) \in c_2^k\}$ and $k \neq L_2$, then go to Step 2.k+1. If $\tau(i,c_2^k) = 1$ for all $i \in \{j \in \widehat{I}_{1} \mid (j,s) \in c_2^k\}$ and $k = L_2$, then go to Step 2.End with $\mu_{2,L_2}^{\prime \prime \prime} = \mu_{2,k}^{\prime \prime }$. If $\tau(i,c_2^k) \neq 1$ for some $i \in \{j \in \widehat{I}_{1} \mid (j,s) \in c_2^k\}$ then go to Step 2.k.A.

\item[Step 2.k.A] Take $i \in \{j \in \widehat{I}_{1} \mid (j,s) \in c_2^k\}$ such that $\tau(i,c_2^k) \neq 1$. If $\#\mu_{2,k}^{\prime \prime}(s^2) < q_{s^2}$, then $\mu_{2,k}^{i2} = \mu_{2,k}^{\prime \prime} + (i,s^2)$ with $(i,s^2) \in c_2^k(i)$. If $\#\mu_{2,k}^{\prime \prime}(s^2) = q_{s^2}$, then $\mu_{2,k}^{i2} = \mu_{2,k}^{\prime \prime} + (i,s^2) - (j,s^2)$ with $(i,s^2) \in c_2^k(i)$, $(j,s^2) \in \mu_{2,k}^{\prime \prime}$ and $F_{s^2}(j)>F_{s^2}(l)$ for all $l\in \mu_{2,k}^{\prime \prime}(s^2)$, $l \neq j$. Next, student $i$ leaves school $s^2$ to become unmatched and guaranteeing a free slot at school $s^2$. We reach $\mu_{2,k}^{i2^{\prime}} = \mu_{2,k}^{i2} - (i,s^2)$. Next, if $\tau(i,c_2^k) \neq 2$ and $\#\mu_{2,k}^{\prime \prime}(s^3) < q_{s^3}$, then $\mu_{2,k}^{i3} = \mu_{2,k}^{i2^{\prime}} + (i,s^3)$ with $(i,s^3) \in c_2^k(i)$. If $\tau(i,c_2^k) \neq 2$ and $\#\mu_{2,k}^{\prime \prime}(s^3) = q_{s^3}$, then $\mu_{2,k}^{i3} = \mu_{2,k}^{i2^{\prime}} + (i,s^3) - (j,s^3)$ with $(i,s^3) \in c_2^k(i)$, $(j,s^3) \in \mu_{2,k}^{\prime \prime}$ and $F_{s^3}(j)>F_{s^3}(l)$ for all $l\in \mu_{2,k}^{\prime \prime}(s^3)$, $l \neq j$. Next, student $i$ leaves school $s^3$ to become unmatched and guaranteeing a free slot at school $s^3$. We reach $\mu_{2,k}^{i3^{\prime}} = \mu_{2,k}^{i3} - (i,s^3)$. We repeat this process until we reach $\mu_{2,k}^{i\tau(i,c_2^k)^{\prime}} = \mu_{2,k}^{i\tau(i,c_2^k)} - (i,s^{\tau(i,c_2^k)})$.

We repeat the process of Step 2.k.A with each student $i \in \{j \in \widehat{I}_{1} \mid (j,s) \in c_2^k\}$ such that $\tau(i,c_2^k) \neq 1$ to reach in the end the matching $\mu_{2,k}^{\prime \prime \prime}$ where all students involved in $c_2^k$ are unmatched and each school $s$ involved in $c_2^k$ has $\#\{(i,s^\prime) \in c_2^k \mid s^\prime = s \}$ free slots. If $k \neq L_2$, then go to Step 2.$k$+1. Otherwise, go to Step 2.End with $\mu_{2,L_2}^{\prime \prime \prime} = \mu_{2,k}^{\prime \prime \prime}$.

\item[Step 2.End.] We have reached $\mu_{2,L_2}^{\prime \prime \prime}$ where each student $i$ involved in $C_2$ is either matched to $\mu^E(i)$ or unmatched and each school $s$ involved in $C_2$ has $\#\{(i,s^\prime) \in \cup_{l=1}^{L_2}c_2^l \mid s^\prime = s \text{ and } \mu_{2,L_2}^{\prime \prime \prime}(i) \neq \mu^E(i) \}$ free slots. Next, those unmatched students join the school they point to in $C_2$ to form the matching $\widetilde{\mu}_2 = \mu_{2,L_2}^{\prime \prime \prime} + \{(i,s) \in M_2 \mid (i,s) \notin \mu_{2,L_2}^{\prime \prime \prime} \}$. Notice that $(M_1 \cup M_2)\subseteq \widetilde{\mu}_2$. If $\widetilde{\mu}_2 = \mu^E$ then the process ends. Otherwise, go to Step 3.1.

\item[End.]The process goes on until we reach $\widetilde{\mu}_{\bar{k}} = \cup_{k=1}^{\bar{k}}M_{k} = \mu^E$.

\end{itemize}
\begin{flushright}
\qedsymbol
\end{flushright}

\section*{References}

\begin{description}

\item[Abdulkadiro\u{g}lu, A., and T. Andersson (2022),] \textquotedblleft School choice\textquotedblright , Working Paper, National Bureau of Economic Research (No. W29822).

\item[Abdulkadiro\u{g}lu, A., Y.K. Che, P. A. Pathak, A.E. Roth, and O. Tercieux (2020),]
\textquotedblleft Efficiency, justified envy, and incentives in priority-based matching\textquotedblright , \textit{American Economic Review Insights} 2, 425-42.

\item[Abdulkadiro\u{g}lu, A., and T. Sönmez (2003),] \textquotedblleft School choice: A mechanism design approach\textquotedblright , \textit{American Economic Review} 93, 729-747.






\item[Che, Y.K., and O. Tercieux (2019),] \textquotedblleft Efficiency and stability in large matching markets\textquotedblright , \textit{Journal of Political Economy} 127, 2301-2342.

\item[Chwe, M. S.-Y. (1994),] \textquotedblleft Farsighted coalitional stability\textquotedblright , \textit{Journal of Economic Theory} 63, 299-325.



\item[Do\u{g}an, B., and L. Ehlers (2021),] \textquotedblleft Minimally unstable Pareto improvements over deferred acceptance\textquotedblright , \textit{Theoretical Economics} 16, 1249-1279.

\item[Do\u{g}an, B., and L. Ehlers (2022),] \textquotedblleft Robust minimal instability of the top trading cycles mechanism\textquotedblright , \textit{American Economic Journal: Microeconomics} 14, 556-582.


\item[Dutta, B., and R. Vohra (2017),] \textquotedblleft Rational expectations and farsighted stability\textquotedblright , \textit{Theoretical Economics} 12, 1191-1227.

\item[Ehlers, L. (2007),] \textquotedblleft Von Neumann-Morgenstern stable sets in matching problems\textquotedblright ,\ \textit{Journal of Economic Theory} 134, 537-547.


\item[Gale, D., and L.S. Shapley (1962),] \textquotedblleft College admissions and the stability of marriage\textquotedblright , \textit{American Mathematical Monthly} 69, 9-15.

\item[Haeringer, G., and F. Klijn (2009),] \textquotedblleft Constrained school choice\textquotedblright , \textit{Journal of Economic Theory} 144, 1921-1947.

\item[Haeringer, G. (2017),] Market design: auctions and matching, MIT Press, Cambridge, MA.

\item[Hakimov, R., and O. Kesten (2018),] \textquotedblleft The equitable top trading cycles mechanism for school choice\textquotedblright , \textit{International Economic Review} 59, 2219-2258.



\item[Herings, P.J.J., A. Mauleon, and V. Vannetelbosch (2017),] \textquotedblleft Stable sets in matching problems with coalitional sovereignty and path dominance\textquotedblright , \textit{Journal of Mathematical Economics} 71, 14-19.

\item[Herings, P.J.J., A. Mauleon, and V. Vannetelbosch (2020),] \textquotedblleft Matching with myopic and farsighted players\textquotedblright , \textit{Journal of Economic Theory} 190, 105125.

\item[Herings, P.J.J., A. Mauleon, and V. Vannetelbosch (2019),] \textquotedblleft Stability of networks under horizon-$K$ farsightedness\textquotedblright , \textit{Economic Theory} 68, 177-201.

\item[Kesten, O (2010),] \textquotedblleft School choice with consent \textquotedblright , \textit{The Quarterly Journal of Economics} 125, 1297-1348.


\item[Luo, C., A. Mauleon, and V. Vannetelbosch (2021),] \textquotedblleft Network formation with myopic and farsighted player\textquotedblright , \textit{Economic Theory} 71, 1283-1317.


\item[Mauleon, A., V. Vannetelbosch and W. Vergote (2011),] \textquotedblleft Von Neumann - Morgenstern farsightedly stable sets in two-sided matching\textquotedblright , \textit{Theoretical Economics} 6, 499-521.

\item[Morrill, T. (2015),] \textquotedblleft Two simple variations of top trading cycles\textquotedblright , \textit{Economic Theory} 60, 123-140.





\item[Ray, D. and R. Vohra (2015),] \textquotedblleft The farsighted stable set\textquotedblright , \textit{Econometrica} 83, 977-1011.

\item[Ray, D., and R. Vohra (2019),] \textquotedblleft Maximality in the farsighted stable set\textquotedblright , \textit{Econometrica} 87, 1763-1779.

\item[Reny, P.J. (2022),] \textquotedblleft Efficient matching in the school choice problem\textquotedblright , \textit{American Economic Review} 112, 2025-43.

\item[Roth, A.E. (1982),] \textquotedblleft The economics of matching: stability and incentives\textquotedblright , \textit{Mathematics of Operations Research} 7, 617-628.

\item[Roth, A.E. and M.A.O. Sotomayor (1990),] Two-sided matching, a study in game-theoretic modeling and analysis, Econometric Society Monographs No.18, Cambridge University Press, Cambridge, UK.


\item[Shapley, L.S., and H. Scarf (1974),] \textquotedblleft On cores and indivisibility\textquotedblright , \textit{Journal of Mathematical Economics} 1, 23-37.


\end{description}

\end{document}